\documentclass[12pt,reqno]{amsart}
\usepackage{amssymb, amsmath}
\usepackage{amsfonts}
\usepackage[margin=1.0in]{geometry}
\usepackage{esint}
\usepackage{enumitem}
\usepackage{latexsym}

\allowdisplaybreaks 
\theoremstyle{plain}

\newcommand{\C}{\mathbb C}

\newtheorem{thm}{Theorem}[section]
\newtheorem{prop}[thm]{Proposition}
\newtheorem{lem}[thm]{Lemma}
\newtheorem{cor}[thm]{Corollary}
\newtheorem{defn}[thm]{Definition}
\newtheorem{fact}[thm]{Fact}

\newtheorem{example}[thm]{Example}

\renewcommand{\proof}{{\em Proof\/}: }
\renewcommand{\qed}{\hfill $\fbox{}$}

\makeatletter\@addtoreset{equation}{section} \makeatother

\DeclareMathOperator{\Tr}{Tr}
\DeclareMathOperator{\vol}{vol}
\DeclareMathOperator{\ran}{ran}

\newcommand{\lb}{\label}
\newcommand{\be}{\begin{equation}}
\newcommand{\ee}{\end{equation}}
\newcommand{\ov}{\overline}
\newcommand{\dd}{\,d}
\newcommand{\mc}{\mathcal}
\newcommand{\A}{\mathcal A}
\newcommand{\B}{\mathcal B}
\newcommand{\K}{\mathcal K}
\begin{document}

\title{Pluriharmonic solutions to Yang-Mills equations: a $C^*$-algebras approach}
\author{Marius Beceanu, Sachin Munshi$^*$, and Rongwei Yang}
\thanks{$^*$ Corresponding author}
\address[Marius Beceanu]{New York State Department of Environmental Conservation, Albany, NY 12233, U.S.A.}
\email{mbeceanu@gmail.com}

\address[Sachin Munshi]{Department of Mathematics, Grand Canyon University (College of Humanities and Social Sciences), Phoenix, AZ 85017, U.S.A.}
\email{Sachin.Munshi@my.gcu.edu, sacmun86@gmail.com}

\address[Rongwei Yang]{Department of Mathematics and Statistics, University at Albany, the State University of New York, Albany, NY 12222, U.S.A.}
\email{ryang@albany.edu}


\date{}

\maketitle

\begin{abstract}
This partially expository paper provides a view of Yang-Mills equations from the perspective of complex variables, operator theory, and $C^{*}$-algebras. Through operator-valued pluriharmonic and skew-Hermitian differential forms, it constructs a new class of instanton solutions. Furthermore, it provides a complex variable version of the Yang-Mills Lagrangian and the Belavin-Polyakov-Schwartz-Tyupkin instanton. 
\end{abstract}

\section{Introduction}

Yang-Mills equations, first put forth in \cite{YM54}, can be viewed as a non-abelian generalization of Maxwell's equations, where the gauge groups can be $\mathbf{SU}\left(2\right) \times \mathbf{U}\left(1\right)$, $\mathbf{SU}\left(3\right)$, and some other related nonabelian compact Lie groups. Solutions to Yang-Mills equations have been discovered in \cite{ablowitz1993self,ADHM,belavin1975pseudoparticle,belavin1979quantum,WY68} and several other papers. In physics, Yang-Mills theory unified electroweak theory with quantum chromodynamics (associated with the strong force), and it is the foundation for the Standard Model in particle physics \cite{CL,JW}.  Since the 1980s, mathematical physicists and mathematicians garnered deep interest in Yang-Mills theory and its applications to the study of $4$-dimensional manifolds and low dimensional topology \cite{baez1994gauge,Don83,Don90,FU84,Wi94}.

Classical Yang-Mills equations are typically defined in the following manner \cite{ablowitz1993self,munshi2020maxwell}. 
Consider a vector bundle $E$ over the Euclidean space $\mathcal{M}=\mathbb{R}^{4}$, with fibers isomorphic to
the Lie algebra $\mathfrak{g}$ with gauge group $G$. In standard coordinates $\mathbf{x}=(x^{\mu}), 0\leq \mu \leq 3$, consider the $\mathfrak{g}$-valued {\em connection}
\[
A=\sum_{\mu=0}^{3}A_{\mu}(\mathbf{x})dx^{\mu}.
\]
The associated {\em curvature field} $F= dA + A\wedge A$ can be written as
\[
F=\sum_{\mu,\nu=0}^{3}F_{\mu\nu}(\mathbf{x})dx^{\mu} \wedge dx^{\nu}.
\]
The metric tensor of $\mathbb{R}^{4}$ is the identity matrix $\text{diag}\left(1,1,1,1\right)=(\delta_{\mu\nu})=(\delta^{\mu\nu})$. Now, classical Yang-Mills equations are given by
 \begin{subequations}
\begin{eqnarray}
\partial_{\mu}A_{\nu}-\partial_{\nu}A_{\mu}-[A_{\mu},A_{\nu}] & = & F_{\mu\nu},\label{eq:1.1a}\\
\partial_{\mu}F^{\mu\nu}-[A_{\mu},F^{\mu\nu}] & = & J^{\nu},\label{eq:1.1b}
\end{eqnarray}
\end{subequations} 
where the $\mathfrak{g}$-valued coefficients $A_{\mu}$ are called {\em Yang-Mills potentials},  $F_{\mu\nu}$ are called {\em Yang-Mills fields} with $F^{\mu\nu}=\delta^{\mu\alpha}\delta^{\nu\beta}F_{\alpha\beta}$, and $J^{\nu}$ is referred to as a non-abelian {\em current}. This paper provides an alternative approach to Yang-Mills equations in terms of complex differential forms, operator theory, and $C^{*}$-algebras. Interestingly, this leads naturally to a new class of instanton solutions to Yang-Mills equations. To make this approach transparent, a sufficient amount of preparation is made before the discussion. The paper is organized as follows. 

\tableofcontents

\section{Complex Differential Forms and the Hodge Star Operator}

We first introduce some preliminaries on complex differential forms. Doing so will allow us to define Maxwell's equations and Yang-Mills equations in terms of these differential forms. Much of the information that is to follow can be found in \cite{baez1994gauge,darling1994differential,flanders1963differential,huybrechts2006complex,lindell2004differential,tu2017differential,westenholz1981differential}, as well as the references therein. In particular, similar preparation is made in \cite{munshi2020maxwell,munshi2022complex} for the discussion on Maxwell's equations.

Consider $\mathbb{C}^{n}$ with points given by coordinates $z=\left(z_{1},z_{2},\dots,z_{n}\right)$.
The {\em tangent space} of $\mathbb{C}^{n}$ is $$T\left(\mathbb{C}^{n}\right)=\text{span}\left\{ \frac{\partial}{\partial z_{k}},\frac{\partial}{\partial\bar{z}_{k}}:1\leq k\leq n\right\}, $$ and 
the {\em cotangent space} is given by $$T^{\ast}\left(\mathbb{C}^{n}\right)=\text{span}\left\{ dz_{k},d\bar{z}_{k}:1\leq k\leq n\right\}. $$ 
\begin{defn}
For any smooth function $f$ on a domain $\mathcal{M}\subset\mathbb{C}^{n}$, consider the linear operators $\partial,\bar{\partial},d$ defined
to act on $f$ as follows:
\begin{align*}
\partial f  =\sum_{k=1}^{n}\frac{\partial f}{\partial z_{k}}dz_{k},\hspace{0.5cm}
\bar{\partial}f =\sum_{k=1}^{n}\frac{\partial f}{\partial\bar{z}_{k}}d\bar{z}_{k},\hspace{0.5cm}
df =\left(\partial+\bar{\partial}\right)f.
\end{align*}
Here, $d$ is called the complex exterior differential operator,
while $\partial,\bar{\partial}$ are called the Dolbeault operators. 
\end{defn}

A smooth function $f$ on a domain $\mathcal{M}$ is said to be {\em holomorphic} if $\bar{\partial} f=0$ everywhere on $\mathcal{M}$. In fact, this implies that $f$ is analytic in each variable.
Note that if $f$ is holomorphic, then $df=\partial f$. The following fact is well-known.
\begin{fact}
$\partial^{2}=\bar{\partial}^{2}=d^{2}=0$.
\end{fact}
For a multi-index $I=(i_1,i_2\cdots, i_p)$, we assume $i_1<i_2<\cdots <i_p$ and define its length $|I|=p$.
\begin{defn}
The space of complex $\left(p,q\right)$-forms on $\mathcal{M}$ is defined as
$$\Omega^{p,q}\left(\mathcal{M}\right):=\bigg\{ \sum_{|I|=p,|J|=q}f_{I,J}dz_{I}\wedge d\bar{z}_{J}:\ f_{I,J}\in C^{\infty}\left(\mathcal{M}\right)\bigg\},\ \ \ 0\leq p, q \leq n.$$
\end{defn}
We may drop the ``$\mathcal{M}$'' from the definition for convenience. Clearly, $\Omega^{1,0}$ is the space of complex differential forms containing only the $dz_{k}$ terms, and $\Omega^{0,1}$
is that containing only the $d\bar{z}_{k}$ terms.
Then in terms of the exterior product on differential forms we have 
\begin{equation*}
\Omega^{p,q}= \underbrace{\Omega^{1,0} \wedge \cdots \wedge \Omega^{1,0}}_{p} \wedge
\underbrace{\Omega^{0,1} \wedge \cdots \wedge \Omega^{0,1}}_{q}.
\end{equation*}
Slightly abusing notation, we shall adhere to the following definition throughout this paper.
\begin{defn}
$\Omega^{k}:=\bigoplus_{p+q=k}\Omega^{p,q}$ is the space of
all complex differential forms of total degree $k=p+q$.
\end{defn}

\begin{defn}
For each $p>0$, the forms given by $\eta=\sum_{|I|=p}f_{I}dz_{I}$, where each
 $f_{I}$ is holomorphic, are called holomorphic p-forms,
and they form a holomorphic section of $\Omega^{p,0}$.
\end{defn}

Note that if $\eta$ is holomorphic then $\bar{\partial}\eta=0$. 

\subsection{Hodge Star Operator $\star$}

We briefly go over some basics of Hodge theory withholding any discussion
on topology or manifold theory. We identify $\C^n$ with ${\mathbb R}^{2n}$ with the representation
$z_k=x_{2k-2}+ix_{2k-1},\ k=1, 2, ..., n$. Given a nondegenerate self-adjoint $2n\times 2n$ matrix $g=\left(g_{st}\right)$, it induces a sesquilinear form $\langle \cdot, \cdot \rangle$ on the tangent space $T\left(\mathbb{R}^{2n}\right)$ with complex coefficients by the evaluations
\begin{equation}\label{tangents}
\left\langle \alpha \frac{\partial}{\partial x_{s}}, \beta \frac{\partial}{\partial x_{t}}\right\rangle_g:=\alpha \overline{\beta}g_{st},\ \ \ 0\leq s, t\leq 2n-1,
\end{equation}
where $\alpha$ and $\beta$ are complex numbers. Then on the cotangent space $T^*\left(\mathbb{R}^{2n}\right)$ one has the corresponding sesquilinear form given by
\begin{equation}\label{cotang}
\langle \alpha dx_s, \beta dx_t\rangle_g:=\alpha \overline{\beta}g^{st},\end{equation}
where $(g^{st})=g^{-1}$. Using the fact
\[\frac{\partial}{\partial z_{k}}=\frac{1}{2}\left(\frac{\partial}{\partial x_{2k-2}}-i\frac{\partial}{\partial x_{2k-1}}\right), \ \ \ k= 1, ..., n,\]
one may regard the aforementioned sesquilinear forms in (\ref{tangents}) and (\ref{cotang}) as sesquilinear forms on $T\left(\mathbb{C}^{n}\right)$ and $T^*\left(\mathbb{C}^{n}\right)$, respectively. Further, (\ref{cotang}) can be extended to a sesquilinear form $\langle \cdot, \cdot \rangle_g$ on $\Omega^{p}$ such that for $\eta=\eta_{1}\wedge\cdots \wedge \eta_{p}$ and $\xi=\xi_{1}\wedge\cdots \wedge \xi_{p}$, one has  \[
\langle \eta,\xi \rangle_g= \det \left(\langle \eta_{s},\xi_{t}\rangle_g\right)_{s,t=1}^{p},\ \ \ \eta, \xi \in \Omega^{p}.
\]
One observes that if the matrix $g$ is positive definite, then $\langle \cdot, \cdot \rangle_g$ is an inner product on the set of constant $p$-forms for each $1\leq p\leq 2n$.

\begin{defn}\label{de:hodge}
The Hodge star operator $\star:\Omega^{p}\left(\mathbb{C}^{n}\right)\rightarrow\Omega^{2n-p}\left(\mathbb{C}^{n}\right)$ is an isomorphism with respect to the bilinear form $\langle \cdot, \cdot \rangle_g$ such that for $\eta, \xi \in \Omega^{p}\left(\C^n\right)$ one has
\[\eta\wedge \star\ \bar{\xi}=\langle \eta, \xi \rangle_g \textup{vol}_g,\]
where $\textup{vol}_g=\left(\frac{i}{2}\right)^n\sqrt{|\det g|}dz_1\wedge d\bar{z}_{1}\wedge \cdots \wedge dz_n\wedge d\bar{z}_{n}$ is the volume form.
\end{defn}

For more details about the Hodge star operator in complex variables, we refer the reader to \cite{huybrechts2006complex}. The following examples are needed for subsequent study.

\subsection{Self-dual and Anti-self-dual Forms}

We now focus on the Hodge star operator on the complex differential forms on $\C^2$ with respect to the Euclidean metric and the Minkowski metric.

{\bf 1}. Under the Euclidean metric, we have 
$\star^2\omega=\omega$ for every $\omega\in \Omega^{2}(\mathbb{C}^{2})$.
Self-dual and anti-self-dual forms are the eigenvectors corresponding to the eigenvalues $1$ and $-1$, respectively, of the Hodge star operator on $\C^2$.
\begin{defn}
A differential form $\omega$ is said to be self-dual if
it is equal to its Hodge dual, i.e. $\star\ \omega=\omega$. If $\star\ \omega=-\omega$,
then $\omega$ is said to be anti-self-dual.
\end{defn}

Out of the complex differential forms we considered in the previous subsection, six of them correspond
to the pair $\left(p,q\right)$ such that $p+q=2$. Let $\Omega_{+}^{2},\Omega_{-}^{2}$ denote
the bases of self-dual and anti-self-dual forms, respectively, in
$\Omega^{2}$. Then it can be verified easily that 
\begin{align}
\Omega_{+}^{2} & =\text{span}\left\{ dz_{1}\wedge dz_{2},d\bar{z}_{1}\wedge d\bar{z}_{2},dz_{1}\wedge d\bar{z}_{1}+dz_{2}\wedge d\bar{z}_{2}\right\} ,\label{eq:2.3}\\
\Omega_{-}^{2} & =\text{span}\left\{ dz_{1}\wedge d\bar{z}_{2},dz_{2}\wedge d\bar{z}_{1},dz_{1}\wedge d\bar{z}_{1}-dz_{2}\wedge d\bar{z}_{2}\right\} .\nonumber 
\end{align}

{\bf 2}. Under the Minkowski metric, we have $\star^{2}\omega= -\omega$ for any $\omega\in \Omega^2(\mathbb{C}^{2})$.
Hence $\star$ has eigenvalues $i, -i$. With a bit of abuse of terminology, the eigenspaces corresponding to them are often also called self-dual and anti-self-dual forms, respectively. For consistency we shall also denote them by $\Omega^2_+$ and $\Omega^2_-$, respectively.
Then we have 
\begin{align}
\Omega_{+}^{2} & =\text{span}\left\{ dz_{1}\wedge dz_{2}+idz_2\wedge d\bar{z}_{1}, d{z}_{1}\wedge d\bar{z}_{1}+idz_{2}\wedge d\bar{z}_{2}, d{z}_{1}\wedge d\bar{z}_{2}+id\bar{z}_{1}\wedge d\bar{z}_{2}\right\} ,\label{eq:2.4}\\
\Omega_{-}^{2} & =\text{span}\left\{ dz_{1}\wedge dz_{2}-idz_2\wedge d\bar{z}_{1}, d{z}_{1}\wedge d\bar{z}_{1}-idz_{2}\wedge d\bar{z}_{2}, d{z}_{1}\wedge d\bar{z}_{2}-id\bar{z}_{1}\wedge d\bar{z}_{2}\right\} .\nonumber 
\end{align}
\noindent For details of the computation, we refer the readers to \cite{munshi2022complex}.

Henceforth, throughout this paper, we shall use the abbreviations SD and ASD to refer to differential forms that are self-dual and anti-self-dual, respectively.

\section{$C^*$-algebras, Covariant Derivative, and Covariant Co-derivative}

Here we go over a few basics regarding $C^{*}$-algebras, a subject extensively studied in the areas of operator theory and functional analysis, among others (see, for instance \cite{davidson1996calgebras} and \cite{murphy2014calgebras}). This will be the setting for two subsequent propositions \cite{munshi2020maxwell}, leading to the definitions of covariant derivative and covariant co-derivative. We begin with the following definition.
\begin{defn}
A $C^{*}$-algebra $\mathcal{A}$ is a Banach algebra with a map $a \rightarrow a^{*}$, $\forall a\in \mathcal{A}$ that has the following properties:
\end{defn}
\begin{itemize}
\item It is an {\em involution}, meaning $a^{**} =\left(a^{*}\right)^{*}=a$,
\item $\left(a+b\right)^{*}=a^{*}+b^{*}, \left(ab\right)^{*}=b^{*}a^{*},\ b\in \mathcal{A}$,
\item $\left(\lambda a\right)^{*}=\bar{\lambda}a^{*},\ \lambda\in \mathbb{C}$,
\item $\|a^{*}a\|= \|a\|^2$.
\end{itemize}
Some important examples of $C^{*}$-algebras include the finite-dimensional matrix algebra $\mathbf{M}_{n}\left(\mathbb{C}\right)$, the algebra $B\left(\mathcal{H}\right)$ of bounded linear operators on a complex Hilbert space $\mathcal{H}$, the algebra $K\left(\mathcal{H}\right)$ of compact operators on a separable infinite-dimensional Hilbert space $\mathcal{H}$, and the space $C\left(X\right)$ of complex-valued continuous functions on a locally compact Hausdorff space $X$.

\subsection{$C^{*}$-algebras, Trace, and Inner Product}

In the sequel, we shall always assume that $\mathcal{A}$ is a unital $C^{*}$-algebra.
\begin{defn}
An element $h\in\mathcal{A}$ is said to be positive if $h=a^{*}a \geq 0$ for some $a\in\mathcal{A}$. bounded linear functional $\phi$ on $\mathcal{A}$ is positive if $\phi\left(h\right)\geq 0$ for all $h\geq 0$.
\end{defn}
	
\begin{defn}
A positive linear functional $\phi$ on $\mathcal{A}$ is said to be a state if $\phi\left(I\right)= 1$.
\end{defn}
\begin{defn}
A tracial state $\textup{Tr}$ on $\mathcal{A}$ is a state on $\mathcal{A}$ such that
\begin{equation}
\textup{Tr}\left(ab\right)=\textup{Tr}\left(ba\right),\ \ \ a, b\in \mathcal{A}.
\label{eq:3.1}
\end{equation}
\end{defn}
In particular, the tracial state $\textup{Tr}$ is said to be {\em faithful} if $\textup{Tr}\left(a^{*}a\right)=0$ only if $a=0$. In this case, we can define an inner product on $\mathcal{A}$ by
\begin{equation}
\langle a,b \rangle = \textup{Tr}\left(ab^{*}\right),\ \ \ a,b \in \mathcal{A}.
\label{eq:3.2}
\end{equation}
If $\mu$ and $\eta$ are two $\mathcal{A}$-valued $p$-forms, then a natural extension of (\ref{eq:3.2}) is given by
\begin{equation}
\langle \mu,\eta \rangle \text{vol}=\text{Tr}(\mu \wedge \star(\eta^{*})). \label{eq:3.3}
\end{equation}
Further, if $\mathcal{M}$ is a complex $n$-dimensional manifold, and $\mu$ and $\eta$ are ${\mathcal A}$-valued differential $p$-forms, then an inner product can be defined by
\begin{equation}\label{iproduct}
(\mu, \eta)=\int_M \langle \mu,\eta \rangle \text{vol}=\int_M \text{Tr}(\mu \wedge \star(\eta^{*}))
\end{equation}
whenever the integral is convergent. The $L^2$-space of $\mathcal{A}$-valued differential $p$-forms, denoted by $L^2(\mathcal{M})\otimes \Omega^p(\mathcal{A})$, is the collection of such forms $\mu$ satisfying $\|\mu\|^2=(\mu, \mu)<\infty$. See \cite{douglas2018hermitian} for more details. Now we have the following proposition. % 2.6
\begin{prop}\label{prop3.5}
Let $\mathcal{A}$ be a $C^{*}$-algebra with a faithful tracial state $\textup{Tr}$ and the induced inner product in \textup{(\ref{eq:3.2})}. If $\eta\in \Omega^p(\mathcal{A})$, then 
$(\star \eta)^{*}= \star(\eta^{*})$.
\end{prop}
\proof
For all $\mu,\eta \in \Omega^{p}(\mathcal{M}, \mathcal{A})$, equation (\ref{eq:3.1}) gives
\begin{align}
\text{Tr}(\mu \wedge (\star \eta)^{*}) & = \text{Tr}((\star \eta)^{*} \wedge \mu)(-1)^{p(2n-p)} \nonumber\\
 & =\text{Tr}((\star \eta)^{*} \wedge \star\star \mu) =\langle (\star \eta)^{*},(\star \mu)^{*} \rangle \text{vol}, \label{eq:3.4}
\end{align}
where the second line follows from the fact that $\star\star \mu=(-1)^{p(2n-p)}\mu$.
For any $A,B\in\mathcal{A}$, since $\text{Tr}(A^{*})= \overline{\text{Tr}(A)}$, it follows that
\begin{align*}
\langle A,B \rangle & =\text{Tr}(AB^{*}) =\overline{\text{Tr}((AB^{*})^{*})} \\
 & = \overline{\text{Tr}(B^{**}A^{*})} = \overline{\text{Tr}(A^ {*}B^{**})} = \overline{\langle A^{*}, B^{*} \rangle}.
\end{align*}
In light of (\ref{eq:3.3}) and the fact that $\star$ is an isometry, (\ref{eq:3.4}) implies that
\begin{align*}
\langle (\star \eta)^{*},(\star \mu)^{*} \rangle \text{vol} & = \overline{\langle \star\eta,\star\mu \rangle}\text{vol}\\
 & =\langle \star\mu,\star\eta \rangle \text{vol} =\langle\mu,\eta \rangle \text{vol}.
\end{align*}
Therefore, we have $\text{Tr}(\mu\wedge\star(\eta^{*}))=\text{Tr}(\mu\wedge(\star\eta)^{*})$ for all $\mu$, giving $(\star \eta)^{*}= \star(\eta^{*})$.\qed

\subsection{Covariant Derivative}

For convenience, we drop ``$\mathcal{M}$'' from the notation ``$\Omega^{p}(\mathcal{M}, \mathcal{A})$''. Thus, for a $C^{*}$-algebra $\mathcal{A}$ and an $\mathcal{A}$-module $\mathcal{H}$,
we denote by $\Omega^{p}\left(\mathcal{A}\right)$ (resp. $\Omega^{p}\left(\mathcal{H}\right)$)
the space of $\mathcal{A}$-valued (resp. $\mathcal{H}$-valued) $p$-forms
on $\mathcal{M}$. 
\begin{defn}
Assuming the Einstein summation convention and given $A=A_jdx_j \in\Omega^{1}\left(\mathcal{A}\right)$, the covariant derivative $D_{A}:\Omega^{p}\left(\mathcal{H}\right)\to\Omega^{p+1}\left(\mathcal{H}\right), 0\leq p\leq 2n$ is defined as
\[D_{A}\left(\omega\right)=d\omega+A\wedge\omega,\]
 where $\omega=\omega_{J}dx_{J}\in\Omega^{p}\left(\mathcal{H}\right), \left|J\right|=p$, and $A\wedge\omega=A_{j}\omega_{J}dx_{j}\wedge dx_{J}$.
\end{defn}
\noindent $D_A$ is often referred to as a {\em connection} on the trivial bundle ${\mathcal M}\times {\mathcal H}$ and $A$ as the associated {\em connection form}.
Note that for any scalar function $\phi$, we have the Leibniz rule:
\[
D_{A}\left(\phi\omega\right)=d\phi\wedge\omega+\phi D_{A}\omega.
\]
 In particular, for $\omega\in\Omega^{p}\left(\mathcal{H}\right)$
one has $D_{A}^{2}\omega=\left(dA+A\wedge A\right)\wedge\omega$.
\begin{defn}
The $2$-form $F_{A}:=dA+A\wedge A$ is said to be the curvature field associated with the connection $D_A$.
\end{defn}

 For $T\in\Omega^{p}\left(\mathcal{A}\right)$ and $\omega\in\Omega^{p}\left(\mathcal{H}\right)$, the product rule implies that
\begin{align}
D_{A}\left(T\wedge\omega\right) & =d\left(T\wedge\omega\right)+A\wedge\left(T\wedge\omega\right) \nonumber\\
 & =dT\wedge\omega+\left(-1\right)^{p}T\wedge d\omega+ A\wedge T\wedge\omega. \label{eq:3.5}
\end{align}
On the other hand, we have
\begin{align}
D_{A}\left(T\wedge\omega\right) &= \left(D_{A}T\right)\wedge\omega +\left(-1\right)^{p}T\wedge D_{A}\omega\nonumber\\
 &= \left(D_{A}T\right)\wedge\omega +\left(-1\right)^{p}T\wedge \left(d\omega+ A\wedge\omega\right). \label{eq:3.6}
\end{align}
Equating (\ref{eq:3.5}) and (\ref{eq:3.6}), one has
\begin{equation}
D_{A}T=dT+A\wedge T-\left(-1\right)^{p}T\wedge A,\ \ \ T\in\Omega^{p}\left(\mathcal{A}\right). \label{eq:3.7}
\end{equation}

\begin{prop}
The map $D_A: \Omega^{p}\left(\mathcal{A}\right)\to \Omega^{p+1}\left(\mathcal{A}\right), 0\leq p\leq 2n$, defined in \textup{(\ref{eq:3.7})} is a connection on the trivial bundle $\C^2\times {\mathcal A}$.
\end{prop}
\proof The linearity of $D_A$ is easy to see. It only remains to check that $D_A$ obeys the Leibniz rule. Indeed, for $T_1\in \Omega^{p}(\mathcal{A})$ and $T_2\in \Omega^{q}(\mathcal{A})$, we have 
\begin{align*}
D_{A}(T_1\wedge T_2) &= d(T_1\wedge T_2)+A\wedge (T_1\wedge T_2)-\left(-1\right)^{p+q}(T_1\wedge T_2)\wedge A\\
&= dT_1\wedge T_2+(-1)^pT_1\wedge dT_2+A\wedge (T_1\wedge T_2)+(-1)^pT_1\wedge A\wedge T_2\\
& \indent -(-1)^pT_1\wedge A\wedge T_2-\left(-1\right)^{p+q}(T_1\wedge T_2)\wedge A\\
&= \big(dT_1+A\wedge T_1-\left(-1\right)^{p}T_1\wedge A\big)\wedge T_2\\
& \indent +(-1)^pT_1\wedge \big(dT_2+A\wedge T_2-\left(-1\right)^{q}T_2\wedge A\big)\\
&= (D_AT_1)\wedge T_2+(-1)^pT_1\wedge (D_AT_2).
\end{align*}

\subsection{Covariant Co-derivative}

In the remaining part of this paper, we shall focus on the special case ${\mathcal M}=\C^2$. It is well-known that, under the Euclidean metric, the co-differential operator $d^*=-\star d\ \star$ \cite{baez1994gauge,huybrechts2006complex}, in the sense that, given two smooth $L^2$ forms $\alpha\in\Omega^{p}\left(\mathbb{C}\right)$ and $\beta\in\Omega^{p+1}\left(\mathbb{C}\right)$ over $\mathbb{C}^2$, we have $(\alpha, d^*\beta)=(d\alpha, \beta)$. On the other hand, the co-differential operator under the Minkowski metric is known to be $d^*=\star\ d\ \star$. In this subsection, we compute $D^*_A$ for ${\mathcal A}$-valued differential forms with respect to the inner product (\ref{iproduct}). As far as the authors are concerned, the following fact is not found in the literature.

\begin{prop}\label{D*}
Given a $C^*$-algebra ${\mathcal A}$ with a tracial state $\Tr$ and $A\in\Omega^{1}\left(\mathcal{A}\right)$, under the Euclidean metric on $\mathbb{C}^n$ the covariant co-derivative $D_{A}^{*}:=-\star D_{-A^{*}}\star$.
\end{prop}
\proof
Let $\alpha\in\Omega^{p}\left(\mathcal{A}\right)$ and $\beta\in\Omega^{p+1}\left(\mathcal{A}\right)$ be $L^2$-forms over $\C^2$, where $0\leq p\leq 3$. We need to show $(D_A\alpha, \beta)=(\alpha, -\star D_{-A^{*}}\star \beta)$.
First of all, $\star\beta\in\Omega^{3-p}\left(\mathcal{A}\right)$.
Note that $d\star\beta$, $\left(-A^{*}\right)\wedge\star\beta$,
$\star\beta\wedge\left(-A^{*}\right)\in\Omega^{4-p}\left(\mathcal{A}\right)$.
This implies $D_{-A^{*}}\star\beta\in\Omega^{4-p}\left(\mathcal{A}\right)$.
Therefore, we do have that $-\star D_{-A^{*}}\star\beta\in\Omega^{p}\left(\mathcal{A}\right)$. In view of (\ref{iproduct}), it is sufficient to check that 
\begin{equation}
\left\langle D_{A}\alpha,\beta\right\rangle =\left\langle \alpha,-\star D_{-A^{*}}\star\beta\right\rangle \label{eq:3.8}
\end{equation}
on $\C^2$. By (\ref{eq:3.7}), we have 
\begin{equation}
\langle D_{A}\alpha , \beta \rangle= \underbrace{\langle d\alpha, \beta \rangle}_{I_{1}}+ \underbrace{\langle A \wedge \alpha, \beta \rangle}_{I_{2}}- (-1)^{p} \underbrace{\langle \alpha \wedge A, \beta \rangle}_{I_{3}}.
\label{eq:3.9}
\end{equation}We now compute $I_{1},I_{2},I_{3}$ in (\ref{eq:3.9}). Clearly,
$I_{1}=\left\langle \alpha,d^{*}\beta\right\rangle =\left\langle \alpha,-\star d\star\beta\right\rangle $.
Computing $I_{2}$ and $I_{3}$ is not as trivial. First, we make the observation that $\star \star \eta= (-1)^{p(4-p)}\eta= (-1)^{p^{2}}\eta$ for 
$\eta \in\Omega^{p}\left(\mathcal{A}\right)$, and $(\star\beta)^{*}= \star(\beta^{*})$ by Proposition \ref{prop3.5}. Then using the properties of the trace, we have 
\begin{align*}
I_{2}\text{ vol} & =\text{Tr}\left(A\wedge\alpha\wedge\star\beta^{*}\right)\\
 & =-\text{Tr}\left(\alpha\wedge\left(\star\beta\right)^{*}\wedge A\right)\\
 & =-\text{Tr}\left(\alpha\wedge\left(A^{*}\wedge\star\beta\right)^{*}\right)\left(-1\right)^{3-p}\\
 & =\left(-1\right)^{p}\text{Tr}\left(\alpha\wedge\left(A^{*}\wedge\star\beta\right)^{*}\right)\\
 & =\left(-1\right)^{p-p^{2}}\text{Tr}\left(\alpha\wedge\star\star\left(A^{*}\wedge\star\beta\right)^{*}\right)\\
 & =\text{Tr}\left(\alpha\wedge\star\left(\star\left(A^{*}\wedge\star\beta\right)^{*}\right)\right)=\left\langle \alpha,\star\left(A^{*}\wedge\star\beta\right)\right\rangle \text{vol}.
\end{align*}
Similarly, we have 
\begin{align*}
I_{3}\text{ vol} & =\text{Tr}\left(\alpha\wedge A\wedge\left(\star\beta\right)^{*}\right)\\
 & =\text{Tr}\left(\alpha\wedge\left(\star\beta\wedge A^{*}\right)^{*}\right)\left(-1\right)^{3-p}\\
 & =\left(-1\right)^{p+1}\text{Tr}\left(\alpha\wedge\left(\star\beta\wedge A^{*}\right)^{*}\right)\\
 & =\left(-1\right)^{p-p^{2}+1}\text{Tr}\left(\alpha\wedge\star\star\left(\star\beta\wedge A^{*}\right)^{*}\right) =-\left\langle \alpha,\star\left(\star\beta\wedge A^{*}\right)\right\rangle \text{vol}.
\end{align*}
It follows that 
\begin{align*}
\left\langle D_{A}\alpha,\beta\right\rangle  & =\left\langle \alpha,-\star d\star\beta\right\rangle +\left\langle \alpha,\star\left(A^{*}\wedge\star\beta\right)\right\rangle +\left(-1\right)^{p}\left\langle \alpha,\star\left(\star\beta\wedge A^{*}\right)\right\rangle \\
 & =\left\langle \alpha,-\star\left(d\star\beta-\left(A^{*}\wedge\star\beta\right)+\left(-1\right)^{p+1}\left(\star\beta\wedge A^{*}\right)\right)\right\rangle \\
 & =\left\langle \alpha,-\star\left(d\star\beta+\left(\left(-A^{*}\right)\wedge\star\beta\right)-\left(-1\right)^{2n-p-1}\left(\star\beta\wedge\left(-A^{*}\right)\right)\right)\right\rangle \\
 & =\left\langle \alpha,-\star D_{-A^{*}}\star\beta\right\rangle. 
\end{align*}\qed

Recall that under the Minkowski metric, we have $d^*=\star\ d\ \star$ and for a $p$-form $\eta$ we have $\star\star\ \eta=(-1)^{p(4-p)+1}\eta$. Going through the proof above with these adjustments, we have the following analogous fact. 
\begin{cor}\label{D*M}
With respect to the Minkowski metric on $\C^2$, we have $D_A^*=\star\ D_{-A^*} \star.$
\end{cor}

\subsection{The Lorenz Gauge Condition}
In studying Maxwell's equations, one often considers scalar-valued $1$-forms $\omega(z)=f_1dz_1+f_2dz_2+f_{\bar{1}}dz_{\bar{1}}+f_{\bar{2}}dz_{\bar{2}}$ that satisfy the {\em Lorenz gauge condition} $d^{*}\omega=0$ \cite{munshi2022complex}. Under the Euclidean metric, we have 
\[d^{*}\omega=-2\big(\bar{\partial}_{1}f_{1}+\bar{\partial}_{2}f_{2}+\partial_{1}f_{\bar{1}}+\partial_{2}f_{\bar {2}}\big),\] and under the Minkowski metric, 
\[d^{*}\omega= 2\left(\partial_{1}f_{1}-\bar{\partial}_{2}f_{2}+\bar{\partial}_{1}f_{\bar{1}}-\partial_{2}f_{\bar{2}}\right).\]
A quantized version of the Lorenz gauge condition is defined as follows.
\begin{defn}
A $1$-form $A\in \Omega^1({\mathcal A})$ satisfies the Lorenz gauge condition if $D_A^*A=0$.
\end{defn}

By formula (\ref{eq:3.7}), the covariant co-derivative $D_{A}^{*}$ of $A\in\Omega^{1}\left(\mathcal{A}\right)$ can be computed as
\begin{align}
D_{A}^{*}A &= -\star \left(D_{-A^{*}}\left(\star A\right)\right) \nonumber \\
&= -\star \left(d\left(\star A\right) +\left(-A^{*}\right)\wedge \left(\star A\right)+ \left(\star A\right)\wedge \left(-A^{*}\right)\right). \label{eq:5.18}
\end{align}
Working in the Euclidean metric, we have
\begin{equation*}
\star A= \frac{1}{2}\left(A_{1}dz_{1}\wedge dz_{2}\wedge d\bar{z}_{2}- A_{2}dz_{1}\wedge dz_{2}\wedge d\bar{z}_{1}+ A_{\overline{1}}dz_{2}\wedge d\bar{z}_{1}\wedge d\bar{z}_{2}- A_{\overline{2}}dz_{1}\wedge d\bar{z}_{1}\wedge d\bar{z}_{2}\right).
\end{equation*}
Applying the exterior derivative $d$ to the above, it follows that
\begin{equation}
d\left(\star A\right)= \frac{1}{2}\sum_{j=1, 2}\left(\bar{\partial}_{j} A_{j}+ \partial_{j} A_{\overline{j}}\right)dz_{1}\wedge dz_{2}\wedge d\bar{z}_{1}\wedge d\bar{z}_{2}.
\label{eq:5.19}
\end{equation}
Observing that $-A^{*}= -A_{1}^{*}d\bar{z}_{1}-A_{2}^{*}d\bar{z}_{2}-A_{\overline{1}}^{*}dz_{1}-A_{\overline{2}}^{*}dz_{2}$, we also have
\begin{align}
-A^{*}\wedge \left(\star A\right) &= -\frac{1}{2}\left(A_{1}^{*}A_{1}+ A_{2}^{*}A_{2}+ A_{\overline{1}}^{*}A_{\overline{1}}+ A_{\overline{2}}^{*}A_{\overline{2}}\right)dz_{1}\wedge dz_{2}\wedge d\bar{z}_{1}\wedge d\bar{z}_{2}, \label{eq:5.20} \\
\star A \wedge \left(-A^{*}\right) &= \frac{1}{2}\left(A_{1}A_{1}^{*}+ A_{2}A_{2}^{*}+ A_{\overline{1}}A_{\overline{1}}^{*}+ A_{\overline{2}}A_{\overline{2}}^{*}\right)dz_{1}\wedge dz_{2}\wedge d\bar{z}_{1}\wedge d\bar{z}_{2}. \label{eq:5.21}
\end{align}
Denoting the sum of (\ref{eq:5.19}), (\ref{eq:5.20}), and (\ref{eq:5.21}), which is also the expression in parentheses from the third line of (\ref{eq:5.18}), by $W$, we have
\begin{equation}
W= \frac{1}{2}\bigg(\sum_{j=1, 2}(\bar{\partial}_{j} A_{j}+\partial_{j} A_{\overline{j}})+ \sum_{j=1, 2}([A_{j}, A_{j}^{*}]+[A_{\overline{j}}, A_{\overline{j}}^{*}])\bigg)dz_{1}\wedge dz_{2}\wedge d\bar{z}_{1}\wedge d\bar{z}_{2}. 
\label{eq:5.22}
\end{equation}
Then it follows that, under the Euclidean metric, we have
\begin{equation}
D_{A}^{*}A= -\star W= -2\bigg(\sum_{j=1, 2}(\bar{\partial}_{j} A_{j}+\partial_{j} A_{\overline{j}})+ \sum_{j=1, 2}([A_{j}, A_{j}^{*}]+[A_{\overline{j}}, A_{\overline{j}}^{*}])\bigg).
\label{eq:5.23}
\end{equation}

Now under the Minkowski metric, we have the following:
\begin{equation*}
\star A=  \frac{1}{2}\left(A_{1}dz_{2}\wedge d\bar{z}_{1}\wedge d\bar{z}_{2}+ A_{2}dz_{1}\wedge dz_{2}\wedge d\bar{z}_{1}+ A_{\overline{1}}dz_{1}\wedge dz_{2}\wedge d\bar{z}_{2}+ A_{\overline{2}}dz_{1}\wedge d\bar{z}_{1}\wedge d\bar{z}_{2} \right).
\end{equation*}
In addition, we have
\begin{equation}
d\left(\star A\right)= \frac{1}{2}\left(\partial_{1} A_{1}- \bar{\partial}_{2} A_{2}+ \bar{\partial}_{1} A_{\overline{1}}- \partial_{2} A_{\overline{2}}\right)dz_{1}\wedge dz_{2}\wedge d\bar{z}_{1}\wedge d\bar{z}_{2},
\label{eq:5.24}
\end{equation}
and also
\begin{align}
-A^{*}\wedge \left(\star A\right) &= -\frac{1}{2}\left(A_{\overline{1}}^{*}A_{1}- A_{2}^{*}A_{2}+ A_{1}^{*}A_{\overline{1}}- A_{\overline{2}}^{*}A_{\overline{2}}\right)dz_{1}\wedge dz_{2}\wedge d\bar{z}_{1}\wedge d\bar{z}_{2}, \label{eq:5.25} \\
\star A \wedge \left(-A^{*}\right) &= -\frac{1}{2}\left(-A_{1}A_{\overline{1}}^{*}+ A_{2}A_{2}^{*}- A_{\overline{1}}A_{1}^{*}+ A_{\overline{2}}A_{\overline{2}}^{*}\right)dz_{1}\wedge dz_{2}\wedge d\bar{z}_{1}\wedge d\bar{z}_{2}. \label{eq:5.26}
\end{align}
Denoting the sum of (\ref{eq:5.24}), (\ref{eq:5.25}), and (\ref{eq:5.26}) by $W'$, we have
\begin{align}
W'= \frac{1}{2}&\big(\partial_{1} A_{1}- \partial_{2} A_{\overline{2}}+ \bar{\partial}_{1} A_{\overline{1}}- \bar{\partial}_{2} A_{2}+ [A_{1},A_{\overline{1}}^{*}]- [A_{2}, A_{2}^{*}]\nonumber \\
&+ [A_{\overline{1}},A_{1}^{*}]- [A_{\overline{2}}, A_{\overline{2}}^{*}]\big)dz_{1}\wedge dz_{2}\wedge d\bar{z}_{1}\wedge d\bar{z}_{2}. \label{eq:5.27}
\end{align}
Therefore, 
\begin{align}
D_{A}^{*}A=-\star W'= 2\big(&\partial_{1} A_{1}- \partial_{2} A_{\overline{2}}+ \bar{\partial}_{1} A_{\overline{1}}- \bar{\partial}_{2} A_{2}+ \left[A_{1},A_{\overline{1}}^{*}\right]- \left[A_{2}, A_{2}^{*}\right]\nonumber \\
&+ \left[A_{\overline{1}},A_{1}^{*}\right]- \left[A_{\overline{2}}, A_{\overline{2}}^{*}\right]\big),
\label{eq:5.28}
\end{align}
since $\star\left(dz_{1}\wedge dz_{2}\wedge d\bar{z}_{1}\wedge d\bar{z}_{2}\right)= -4$ in the Minkowski metric.

An element $a\in \mathcal{A}$ is said to be {\em normal} if $a^*a=aa^*$. The Fuglede theorem \cite{Fu} asserts that in this case an element $b\in \mathcal{A}$ commutes with $a$ if and only if it commutes with $a^*$.
\begin{defn}
A connection form $A(z)=\sum_{j=1}^2\left(A_j(z)dz_j+A_{\bar{j}}(z)d\bar{z}_j\right) \in \Omega^1(\mathcal{A})$ is said to be normal if $A_j(z)$ and $A_{\bar{j}}(z)$ are normal for each $z\in \C^2$ and $j=1, 2$.
\end{defn}

\begin{lem}\label{lem:comm}
If $f:\C^2\to {\mathcal A}$ is holomorphic and its range $\ran (f)$ consists of normal elements, then $\ran (f)$ is a commuting set.
\end{lem}
\proof Since $f$ is normal, we have $f(z)f^*(z)=f^*(z)f(z)$ on $\C^2$. It follows that
\[\frac{\partial^mf(z)}{\partial z_1^j\partial z_2^k}f^*(z)=f^*(z)\frac{\partial^mf(z)}{\partial z_1^j\partial z_2^k},\ \ \ j, k\geq 0, m=j+k.\]Fixing any $z, z'\in \C^2$ and writing $f(z')$ as a Taylor series in $(z'_i-z_i), i=1, 2$ based at $z$, we obtain $f(z')f^*(z)=f^*(z)f(z')$. Since $f(z)$ is normal, the Fuglede theorem gives $f(z')f(z)=f(z)f(z')$.\qed

When the connection form $A$ is normal, all the commutators in (\ref{eq:5.23}) vanish, and the quantized Lorenz gauge condition $D_{A}^{*}A=0$ looks the same as that for Maxwell's equations. Furthermore, the following corollary is immediate.
\begin{cor}
If a normal connection form $A$ is such that $A_1, A_2$ are holomorphic and $A_{\bar{1}}, A_{\bar{2}}$ are conjugate holomorphic, then it satisfies the Lorenz gauge condition under the Euclidean metric.
\end{cor}
\noindent The following is a direct consequence of (\ref{eq:5.23}).

\begin{cor}\label{LorenzM}
If $A$ is a connection form such that $A_{\overline{j}}=A_{j}^{*}, j= 1, 2$, then in the Minkowski metric we have
$D_{A}^{*}A= 2\left(\partial_{1} A_{1}- \partial_{2} A_{2}^{*}+ \bar{\partial}_{1} A_{1}^{*}- \bar{\partial}_{2} A_{2}\right)$. In particular, connection $A$ satisfies the Lorenz gauge condition if $\partial_{1} A_{1}- \partial_{2} A_{2}^{*}=0$.
\end{cor}
\noindent The second claim in Corollary \ref{LorenzM} holds because $\bar{\partial}_{1} A_{1}^{*}- \bar{\partial}_{2} A_{2}=\big(\partial_{1} A_{1}- \partial_{2} A_{2}^{*}\big)^*$. 

The two corollaries above make it easy to construct connection forms that satisfy the Lorenz gauge condition. More importantly, they point us to a promising direction to look for new solutions to Yang-Mills equations-an exploration undertaken in Section 5. 

\section{Yang-Mills Equations}

Maxwell's equations provide a description of electromagnetic dynamics. They are widely regarded as a foundation of modern physics. Yang-Mills equations are a quantization of Maxwell's equations, and the quantization process is done elegantly in the differential form formulation introduced in \cite{flanders1963differential}. In \cite{munshi2022complex}, the authors made an exploration of Maxwell's equations based on complex differential forms, which led to the discovery of a class of harmonic and pluriharmonic solutions. It is thus a tempting question whether a similar exploration may be fruitful for Yang-Mills equations.
Chapters 4 and 5 aim to address this question.

\subsection{A Quantization of Maxwell's Equations} 

Given a real-valued $1$-forms $\omega$ over ${\mathbb R}^4$, we set $F_\omega=d\omega$. In the language of differential forms, Maxwell's equations are formulated as follows:
\begin{eqnarray}
dF_\omega = 0, \hspace{1cm} \star \text{ } d \star F_\omega  =  J,\label{eq:4.1}
\end{eqnarray}
where \begin{align}
F_\omega=d\omega & =-dx_{0}\wedge\left(E_{1}dx_{1}+E_{2}dx_{2}+E_{3}dx_{3}\right)\label{eq:4.2}\\
 & -B_{1}dx_{2}\wedge dx_{3}+B_{2}dx_{1}\wedge dx_{3}-B_{3}dx_{1}\wedge dx_{2}\nonumber 
\end{align}
is called the {\em Faraday} $2$-{\em form}. Here $\omega$ and $J$ are the {\em potential} $1$-{\em form} and resp. {\em current} $1$-{\em form}, $d$ is the exterior differential operator, and $\star$ is the Hodge star operator with respect to the Minkowski metric. The vectors ${\bf E}=(E_1, E_2, E_3)$ and ${\bf B}=(B_1, B_2, B_3)$ are the corresponding {\em electric} and {\em magnetic fields}. The {\em Lagrangian} for Maxwell's equations is given by 
\begin{align}\label{eq:4.3lag}
L(\omega)=-\omega\wedge \star J+\frac{1}{2}F_\omega\wedge \star F_\omega.
\end{align}
It is a well-known fact that $F_\omega$ is a solution to (\ref{eq:4.1}) if and only if $\omega$ gives an extremal value to the functional $h(\omega)=\int_{{\mathbb R}^4}L(\omega).$ 

Now we do the quantization of (\ref{eq:4.1}) and (\ref{eq:4.3lag}) in the setting where $\mathcal{A}$ is a unital $C^{*}$-algebra with a tracial state $\Tr$. Consider $\mathcal{A}$-valued smooth functions $A_{j}(z), A_{\overline{j}}(z), j=1, 2, z\in \mathbb{C}^2$, and let $\frac{\partial}{\partial z_{j}},\frac{\partial}{\partial \bar{z}_{j}}$ be denoted by $\partial_{j},\bar{\partial}_{j}$, respectively. A connection form $A\in \Omega^{1} (\mathcal{A})$ can be expressed as follows:
\begin{equation}
A= A_{1}dz_{1}+A_{2}dz_{2}+A_{\overline{1}}d\bar{z}_{1}+A_{\overline{2}}d\bar{z}_{2}.
\label{eq:4.4}
\end{equation}

\begin{defn}
Yang-Mills equations are given by 
\begin{equation}
D_{A}F_{A}=0,\text{ \text{ \text{ }}}D_{A}^{*}F_{A}=J,\label{eq:4.5}
\end{equation}
 where $J\in\Omega^{1}\left(\mathcal{A}\right)$ is referred to as a current potential form, and $D_{A}^{*}$ is as defined in \textup{Proposition \ref{D*}} or \textup{Corollary \ref{D*M}} depending on the metric. 
\end{defn}
\noindent Note that in vacuum $J=0$, meaning that the equations are source-free.
Yang-Mills equations can be viewed as a generalization (or a quantization) of Maxwell's
equations, where the scalar field $\C$ is replaced by a possibly non-commutative $C^*$-algebra $\mathcal{A}$. 

For simplicity, we write the curvature field $F_A=dA+A\wedge A$ as
\begin{equation}
F_A=F_{12}dz_{1}\wedge dz_{2}+ F_{\overline{1}\overline{2}}d\bar{z}_{1}\wedge d\bar{z}_{2}+ \sum_{1 \leq j, k \leq 2}F_{j\overline{k}}dz_{j}\wedge d\bar{z}_{k},
\label{eq:4.6}
\end{equation}
where
\begin{align}
F_{12} &= \partial_{1} A_{2}- \partial_{2} A_{1} +\left[A_{1}, A_{2}\right], \nonumber \\
F_{\overline{1}\overline{2}} &= \bar{\partial}_{1} A_{\overline{2}} -\bar{\partial}_{2} A_{\overline{1}} +\left[A_{\overline{1}}, A_{\overline{2}}\right], \label{eq:4.7} \\
F_{j\overline{k}} &= \partial_{j} A_{\overline{k}}- \bar{\partial}_{k} A_{j}+ \left[A_{j}, A_{\overline{k}}\right], \nonumber
\end{align}
for $1 \leq j, k \leq 2$. Since $F_{A}$ is a $2$-form, from (\ref{eq:3.7}) we have that
\begin{equation}
D_{A}F_{A}= dF_{A}+ A\wedge F_{A} - F_{A} \wedge A,
\label{eq:4.8}
\end{equation}
and it follows from the Bianchi identity that $D_{A}F_{A}=0$. The Hodge dual of $F_{A}$ is also a $2$-form, and under the Euclidean metric, it is given by
\begin{align}
\star F_{A} &= F_{12}dz_{1} \wedge dz_{2}+ F_{\overline{1}\overline{2}}d\bar{z}_{1}\wedge d\bar{z}_{2} + F_{2\overline{2}}dz_{1}\wedge d\bar{z}_{1} \label{eq:4.9} \\
&+ F_{1\overline{1}}dz_{2}\wedge d\bar{z}_{2} - F_{1\overline{2}}dz_{1} \wedge d\bar{z}_{2}- F_{2\overline{1}}dz_{2}\wedge d\bar{z}_{1}. \nonumber
\end{align}
Then from Theorem 3.9, we have that
\begin{equation}
D_{A}^{*}F_{A}= -\star D_{-A^{*}}\star F_{A}= -\star \left(d\star F_{A} +\left(-A^{*}\right) \wedge \star F_{A}- \star F_{A} \wedge \left(-A^{*}\right)\right).
\label{eq:4.10}
\end{equation}
Using our knowledge of the exterior differential operator $d$, the wedge product, and the Hodge star operator in the Euclidean metric, after some algebra, simplification, and rearrangement of terms, we have that
\begin{equation}
D_{A}^{*}F_{A}= \underbrace{J_{1}dz_{1} +J_{2}dz_{2}+ J_{\overline{1}}d\bar{z}_{1}+ J_{\overline{2}}d\bar{z}_{2}}_{J},
\label{eq:4.11}
\end{equation}
where the $C^{*}$-algebraic coefficients $J_{j}, J_{\overline{j}}, j= 1, 2$ of the above Yang-Mills non-abelian current $J$ are given by
\begin{align}
J_{1} &= 2\left(\partial_{1}F_{1\overline{1}}+ \partial_{2}F_{1\overline{2}}+ \bar{\partial}_{2}F_{12}- \left[A_{2}^{*},F_{12}\right]- \left[A_{\overline{1}}^{*},F_{1\overline{1}}\right]- \left[A_{\overline{2}}^{*},F_{1\overline{2}}\right]\right), \nonumber \\
J_{2} &= -2\left(-\partial_{1}F_{2\overline{1}}- \partial_{2}F_{2\overline{2}}+ \bar{\partial}_{1}F_{12}- \left[A_{1}^{*},F_{12}\right]+ \left[A_{\overline{1}}^{*},F_{2\overline{1}}\right]+ \left[A_{\overline{2}}^{*},F_{2\overline{2}}\right]\right), \label{eq:4.12}  \\
J_{\overline{1}} &= 2\left(\partial_{2}F_{\overline{1}\overline{2}}- \bar{\partial}_{1}F_{1\overline{1}}- \bar{\partial}_{2}F_{2\overline{1}}+ \left[A_{1}^{*},F_{1\overline{1}}\right]+ \left[A_{2}^{*},F_{2\overline{1}}\right]- \left[A_{\overline{2}}^{*},F_{\overline{1}\overline{2}}\right]\right), \nonumber \\
J_{\overline{2}} &= -2\left(\partial_{1}F_{\overline{1}\overline{2}}+ \bar{\partial}_{1}F_{1\overline{2}}+ \bar{\partial}_{2}F_{2\overline{2}}- \left[A_{1}^{*},F_{1\overline{2}}\right]- \left[A_{2}^{*},F_{2\overline{2}}\right]- \left[A_{\overline{1}}^{*},F_{\overline{1}\overline{2}}\right]\right). \nonumber
\end{align}
On the other hand, under the Minkowski metric, the Hodge dual of $F_{A}$ is given by the $2$-form
\begin{align}
\star F_{A} &=  F_{2\overline{1}}dz_{1} \wedge dz_{2}- F_{1\overline{2}}d\bar{z}_{1}\wedge d\bar{z}_{2} + F_{2\overline{2}}dz_{1}\wedge d\bar{z}_{1} \label{eq:4.13} \\
&- F_{1\overline{1}}dz_{2}\wedge d\bar{z}_{2} +  F_{\overline{1}\overline{2}}dz_{1} \wedge d\bar{z}_{2}- F_{12}dz_{2}\wedge d\bar{z}_{1}. \nonumber
\end{align}
Using (\ref{eq:4.10}), it follows that
\begin{equation}
D_{A}^{*}F_{A}= \underbrace{J_{1}^{'}dz_{1} +J_{2}^{'}dz_{2}+ J_{\overline{1}}^{'}d\bar{z}_{1}+ J_{\overline{2}}^{'}d\bar{z}_{2}}_{J^{'}},
\label{eq:4.14}
\end{equation}
where the $C^{*}$-algebraic coefficients $J_{j}^{'}, J_{\overline{j}}^{'}, j= 1, 2$ of the above Yang-Mills non-abelian current $J^{'}$ are given by
\begin{align}
J_{1}^{'} &= -2\left(-\partial_{2}F_{1\overline{2}}+ \bar{\partial}_{1}F_{1\overline{1}}- \bar{\partial}_{2}F_{12}- \left[A_{1}^{*},F_{1\overline{1}}\right]+ \left[A_{2}^{*},F_{12}\right]+ \left[A_{\overline{2}}^{*},F_{1\overline{2}}\right]\right), \nonumber \\
J_{2}^{'} &= -2\left(-\partial_{1}F_{12}- \partial_{2}F_{2\overline{2}}+ \bar{\partial}_{1}F_{2\overline{1}}- \left[A_{1}^{*},F_{2\overline{1}}\right]+ \left[A_{\overline{1}}^{*},F_{12}\right]+ \left[A_{\overline{2}}^{*},F_{2\overline{2}}\right]\right), \label{eq:4.15}  \\
J_{\overline{1}}^{'} &= -2\left(-\partial_{1}F_{1\overline{2}}- \bar{\partial}_{1}F_{\overline{1}\overline{2}}+ \bar{\partial}_{2}F_{2\overline{2}}+ \left[A_{1}^{*},F_{\overline{1}\overline{2}}\right]- \left[A_{2}^{*},F_{2\overline{2}}\right]+ \left[A_{\overline{1}}^{*},F_{1\overline{2}}\right]\right), \nonumber \\
J_{\overline{2}}^{'} &= -2\left(-\partial_{1}F_{1\overline{2}}- \bar{\partial}_{1}F_{\overline{1}\overline{2}}+ \bar{\partial}_{2}F_{2\overline{2}}+ \left[A_{1}^{*},F_{\overline{1}\overline{2}}\right]- \left[A_{2}^{*},F_{2\overline{2}}\right]+ \left[A_{\overline{1}}^{*},F_{1\overline{2}}\right]\right). \nonumber
\end{align}

Given a connection form $A\in \Omega^1(\mathcal{A})$, we write $F_A$ as in (\ref{eq:4.6}). The following are direct consequences of (\ref{eq:2.3}) and (\ref{eq:2.4}), respectively.
\begin{prop}\label{esdmsd}
With respect to the Euclidean metric,
\begin{align}
F_{A} \textup{ SD} & \Leftrightarrow F_{1\overline{2}}= F_{2\overline{1}}=0, F_{1\overline{1}}=F_{2\overline{2}}, \label{eq:4.18} \\
F_{A} \textup{ ASD} & \Leftrightarrow F_{12}= F_{\overline{1}\overline{2}}=0, F_{1\overline{1}}=-F_{2\overline{2}}. \label{eq:4.19} 
\end{align}

And with respect to the Minkowski metric, we have
\begin{align}
F_{A} \textup{ SD} & \Leftrightarrow F_{2\overline{1}}=iF_{12}, F_{\overline{1}\overline{2}}=iF_{1\overline{2}}, F_{2\overline{2}}=iF_{1\overline{1}}, \label{eq:4.20} \\
F_{A} \textup{ ASD} & \Leftrightarrow F_{2\overline{1}}=-iF_{12}, F_{\overline{1}\overline{2}}=-iF_{1\overline{2}}, F_{2\overline{2}}=-iF_{1\overline{1}}. \label{eq:4.21} 
\end{align}
\end{prop}

Two observations regarding Proposition \ref{esdmsd} and the preceding computation are worth mentioning.

{\bf 1}. We let $GL(\mathcal{A})$ denote the set of invertible elements in ${\mathcal A}$. Any smooth function $g: {\mathcal M}\to GL(\mathcal{A})$ gives rise to a so-called {\em gauge transform} $A\to A_g:= g^{-1}Ag+ g^{-1}dg$. Direct computation gives $F_{A_g}= g^{-1}F_{A}g$, namely,
\begin{equation}\label{gaugeinv}
(F_{A_g})_{i\bar{j}}=g^{-1}(F_{A})_{i\bar{j}}g, \ i, j=1, 2,\ \ \ (F_{A_g})_{12}=g^{-1}(F_{A})_{12}g,\ \ \ (F_{A_g})_{\bar{12}}=g^{-1}(F_{A})_{\bar{12}}g.
\end{equation}

It thus follows directly from Proposition \ref{esdmsd} that, under both the Euclidean and Minkowski metrics, the curvature form $F_{A_g}$ is SD (resp. ASD) if and only if $F_{A}$ is. Two connection forms $A$ and $A'$ are said to be {\em gauge equivalent} if there exists such a function $g$ so that $A_g=A'$. If the values of $g$ are unitary operators, then we say that $A_g$ is a {\em unitary gauge transformation} of $A$. In this case, it can be verified that Yang-Mills equations (\ref{eq:4.5}) are transformed into:
$D_{A_g}F_{A_g}=0,\ D^*_{A_g}F_{A_g}=g^{-1}Jg$.

If the values of $g$ and $A$ commute, we say that $A_g$ is a {\em commuting gauge transformation} of $A$, in which case we have $A_g = A + g^{-1}dg$ and $F_{A_g}= F_{A}$. In addition, if all values of $g$ commute with each other, then $g^{-1}dg = d(\log g)$ is an exact differential (on a simply-connected domain). In particular, $A$ and $A'$ are gauge-equivalent if their difference is an exact form whose values commute with each other and commute point-wise with $A$. We will see such an occurrence in Ch. 6.

%Non-commutative gauge transformations only preserve $F_A$ up to conjugacy.

{\bf 2}. The algebra $\mathcal{A}$ has a natural Lie algebra structure with the bracket defined by the commutator. Given two subsets $S_1, S_2\subseteq \mathcal{A}$, we set $[S_1, S_2]:=\text{span}\{[s_1, s_2]\mid s_i\in S_i, i=1, 2\}$. In the case where $A$ in (\ref{eq:4.4}) is a constant connection form, i.e., $A_i, A_{\bar{i}}, i=1, 2$, are fixed elements in $\mathcal{A}$, then the partial derivatives in (\ref{eq:4.12}) and (\ref{eq:4.15}) vanish, and hence
the components of $J$ and $J'$ belong to $[\mathcal{A}, [\mathcal{A}, \mathcal{A}]]$. A Lie algebra $\mathcal{A}$ is said to be {\em 3-nilpotent} if $[\mathcal{A}, [\mathcal{A}, \mathcal{A}]]=0$. Thus (\ref{eq:4.12}) and (\ref{eq:4.15}) lead to the following observation made in \cite{Sch88}.

\begin{cor}
If $\mathcal{A}$ is a $3$-nilpotent Lie algebra, then every $\mathcal{A}$-valued constant connection form $A$ gives rise to a solution $F_A$ to the source-free Yang-Mills equations with respect to the Euclidean and Minkowski metrics.
\end{cor}
\noindent Furthermore, in view of (\ref{eq:4.7}), if $[\mathcal{A}, \mathcal{A}]\neq 0$ then such a solution $F_A$ can be nontrivial. 

\subsection{The Yang-Mills Lagrangian}

In the $C^*$-algebra setting here, the adjoint operation on ${\mathcal A}$ makes it possible to write down the Lagrangian for Yang-Mills equations in complex variables. Let $A$ and $J$ in $\Omega^{1} (\mathcal{A})$ be the connection form and resp. the current form.
\begin{defn}
The Yang-Mills Lagrangian is defined as 
\[L(A)=-A\wedge \star J^*+\frac{1}{2}F_A\wedge \star F_A^*,\]
and the Yang-Mills functional is defined as 
\[H(A)=\int_{{\mathbb C}^2}\Tr L(A)=(A, -J)+\frac{1}{2} (F_A, F_A)\]
for all $A$ such that the integral is convergent.
\end{defn}

\begin{defn}\label{critical}
A connection form $A\in \Omega^{1} (\mathcal{A})$ is said to be critical if for every $B\in \Omega^{1} (\mathcal{A})$ we have $\frac{d H(A+tB)}{dt}\mid_{t=0}=0$.
\end{defn}
Observe that if $A$ and $J$ are real $1$-forms then $H(A)$ is real, and thus $A$ is critical when the value of $H(A)$ is extremal. When $A$ and $J$ are complex, the functional $H(A)$ can take on complex values, and hence the concept of extremal value is not well-defined in this case. However, Definition \ref{critical} remains valid.

\begin{prop}
A connection form $A\in \Omega^{1} (\mathcal{A})$ satisfies Yang-Mills equations if and only if it is critical.
\end{prop}
\proof Since the equation $D_AF_A=0$ always holds due to the Bianchi identity, we only need to consider the inhomogeneous equation in (\ref{eq:4.5}).
 Let $B\in \Omega^{1} (\mathcal{A})$ be arbitrary and set $A_t=A+tB$, where $t\in \mathbb R$.
Then by direct computation, we have
\begin{align*}
F_{A_t}&=dA_t+A_t\wedge A_t\\
&=F_A+t(dB+A\wedge B+B\wedge A)+t^2B\wedge B\\
&=F_A+tD_AB+t^2B\wedge B.
\end{align*}
Therefore,
\begin{align*}
H(A_t)&=\int_{{\mathbb R}^4}-\text{Tr}A_t\wedge \star J^*+\frac{1}{2}\text{Tr}\left(F_{A_t}\wedge \star F_{A_t}^*\right)\\
&=\int_{{\mathbb C}^2}-\langle A, J\rangle+\frac{1}{2}\langle F_A, F_A\rangle+
     t\left(-\langle B, J\rangle+\frac{1}{2}\langle F_A, D_A B\rangle+\frac{1}{2}\langle D_A B, F_A \rangle+O(t)\right)\vol,
     \end{align*}
where $O(t)$ stands for the remainder term with factor $t$. Since $B$ is arbitrary, up to multiplication by a unimodular scalar, we may assume the integral of $\langle F_A, D_A B\rangle$ is real. Thus 
\[\int_{{\mathbb C}^2}\langle F_A, D_A B\rangle \vol=\int_{{\mathbb C}^2} \langle D_A B, F_A\rangle \vol,\] 
and it follows that
\begin{align*}
\frac{d H(A+tB)}{dt}\mid_{t=0}&=\int_{{\mathbb C}^2}\langle B, -J\rangle+\langle D_AB, F_A\rangle \vol\\
&=\int_{{\mathbb C}^2}\langle B, D_A^*F_A-J\rangle \vol=(B, D_A^*F_A-J).
\end{align*}
It follows that $A$ is critical if and only if $( B, D_A^*F_A-J)=0$ for every $B\in \Omega^{1} (\mathcal{A})$, which is so if and only if $D_A^*F_A=J$.\qed

Observe that in vacuum we have $J=0$, and hence $H(A)\geq 0$ for any $A\in \Omega^{1}\left(\mathcal{A}\right)$. Thus, if $H(A)$ is extremal, then $A$ is a solution to the homogeneous Yang-Mills equations.

\section{Special Types of Connection Forms}

This section aims to look for new solutions to Yang-Mills equations. The discussion in Sections 3.3 and 3.4 lights the path for our approach.

\subsection{Maurer-Cartan Forms and Skew-Hermitian Forms}

 In this paper, a connection form $A\in \Omega^{1}\left(\mathcal{A}\right)$ is said to be a {\em Maurer-Cartan form} if it satisfies the {\em Maurer-Cartan equation}
\begin{equation}\label{MCartan}
F_A=dA+A\wedge A=0,
\end{equation}
\noindent i.e., it gives a flat curvature form. For example, if there exists a smooth function $g:\C^2\to GL({\mathcal A})$ such that $A(z)=g^{-1}(z)dg(z)$, then differentiating
the identity $g^{-1}\left(z\right)g\left(z\right)=I$ yields $d(g^{-1})=-g^{-1}(dg)g^{-1}$, and it follows that $A$ is a Maurer-Cartan form. For a general Maurer-Cartan form $A$, taking the adjoint of (\ref{MCartan}), we obtain\[0=(dA+A\wedge A)^*=dA^*-A^*\wedge A^*.\]Here, the negative sign appears because \[(A_jdz_j\wedge A_kdz_k)^*=A_k^*A_j^*d\bar{z}_j\wedge d\bar{z}_k=-A_k^*d\bar{z}_k\wedge A_j^*d\bar{z}_j.\]
 
The $C^*$-algebra setting also makes the following definition possible.
\begin{defn}
A connection form $A\in\Omega^{1}\left(\mathcal{A}\right)$ is said to
be 
\begin{enumerate}[label=\textup{(\alph*)}, leftmargin=*]
\item Hermitian if $A^{*}=A$, and
\item skew-Hermitian if $A^{*}=-A$.
\end{enumerate}
\end{defn}
The following is easy to check.
\begin{lem}\label{skewH}
Let $A\in\Omega^{1}\left(\mathcal{A}\right)$ be a connection form. Then
\begin{enumerate}[label=\textup{(\alph*)}, leftmargin=*]
\item $A$ is Hermitian if and only if $A=\eta+\eta^*$;
\item $A$ is skew-Hermitian if and only if $A=\eta-\eta^*$, 
\end{enumerate}
where $\eta=A_1dz_1+A_2dz_2\in \Omega^{1}\left(\mathcal{A}\right)$ is uniquely determined by $A$.
\end{lem}
The next three lemmas show some connections among Maurer-Cartan forms, holomorphic forms, and skew-Hermitian forms.
\begin{lem}
A form $\eta=A_1dz_1+A_2dz_2\in  \Omega^{1}\left(\mathcal{A}\right)$ is Maurer-Cartan if and only if it is holomorphic and 
\begin{equation}\label{FF12}
\partial_1A_2-\partial_2A_1+[A_1, A_2]=0.
\end{equation}
\end{lem}
\proof Observes that
\[d\eta+\eta\wedge \eta=\partial \eta+\eta\wedge \eta+\bar{\partial}\eta,\]
in which the first two terms on the right are $(2, 0)$-forms and the third term is a $(1, 1)$-form. Thus, the Maurer-Cartan equation $d\eta+\eta\wedge \eta=0$ holds if and only if $\bar{\partial}\eta=0$ and 
\[\partial \eta+\eta\wedge \eta=(\partial_1A_2-\partial_2A_1+[A_1, A_2])dz_1\wedge dz_2=0,\] or in other words, the form $\eta$ is holomorphic and satisfies (\ref{FF12}).\qed

\noindent Note that the left-hand side of (\ref{FF12}) is $F_{12}$ in (\ref{eq:4.7}). 

\begin{lem}\label{skew}
If $A(z)=f(z)u^{-1}\left(z\right)du\left(z\right), z\in \C^2$, where $f$ is a real-valued scalar function and $u: \C^2\to GL({\mathcal A})$ is differentiable and unitary-valued, then $A$ is skew-Hermitian.
\end{lem}
\proof
It suffices to prove the case $A=u^{-1}\left(z\right)du\left(z\right)$.
Then its adjoint is given by $A^{*}=du^{*}\left(z\right)\left(u^{-1}\left(z\right)\right)^{*}$.
Since $u$ is unitary-valued, $u^{-1}\left(z\right)=u^{*}\left(z\right)$ and thus $A^{*}=du^{-1}\left(z\right)u\left(z\right)$. Differentiating
the identity $u^{-1}\left(z\right)u\left(z\right)=I,$ we have the equivalence
\[
du^{-1}\left(z\right)u\left(z\right)+u^{-1}\left(z\right)du\left(z\right)=0\iff du^{-1}\left(z\right)u\left(z\right)=-u^{-1}\left(z\right)du\left(z\right).
\]
Hence, $A^{*}=-A$ as needed.\qed

\begin{lem} The skew-Hermiticity is unitary gauge invariant. \end{lem}
\proof
Let $A\in\Omega^{1}\left(\mathcal{A}\right)$ be skew-Hermitian and consider $A_u= u^{-1}Au +u^{-1}du$, where $u\in \Omega^0(\mathcal{A})$ is unitary-valued. We need to show that $A_u$ is skew-Hermitian. Indeed,
\begin{align}
\big(A_u\big)^{*} &= \left(u^{-1}Au +u^{-1}du\right)^{*} \nonumber \\
&= u^{*}A^{*}\left(u^{-1}\right)^{*} + \left(du^{*}\right)\left(u^{-1}\right)^{*}\nonumber \\
&= -u^{-1}Au +\left(du^{*}\right)u.  \label{eq:4.16}
\end{align}
Applying $d$ to both sides of the equation $u^{*}u = I$, we have $\left(du^{*}\right)u + u^{*}du =0$, and this implies that 
\begin{equation}
\left(du^{*}\right)u = -u^{*}du = -u^{-1}du. \label{eq:4.17}
\end{equation}
Substituting (\ref{eq:4.17}) into the last line of (\ref{eq:4.16}), we obtain
\begin{equation*}
\big(A_u\big)^{*} = -u^{-1}Au -u^{-1}du = -A_u,
\end{equation*}
which proves the lemma.\qed

The following fact highlights the importance of skew-Hermitian forms.
\begin{prop}\label{SHsolution}
Suppose $A\in \Omega^1(\mathcal{A})$ is skew-Hermitian. 
\begin{enumerate}[label=\textup{(\alph*)}, leftmargin=*]
\item The curvature $F_{A}$ is skew-Hermitian.
\item If $F_{A}$ is SD or ASD, then it is a solution to the source-free Yang-Mills equations.
\end{enumerate}
\end{prop}
\proof
In light of Lemma \ref{skewH} (b), we can write $A=\eta-\eta^*$. It follows that
\begin{align*}
F_A=d\eta -d\eta^* +(\eta-\eta^*)\wedge (\eta-\eta^*)=:\omega-\omega^*,
\end{align*}
where $\omega=(d\eta+\eta\wedge \eta-\eta\wedge \eta^*)$. This proves part (a).

For part (b), under the Euclidean metric, if $F_A$ is SD or ASD then $\star F_{A}=\pm F_{A}$. By Proposition \ref{D*} one has, 
\[D_{A}^{*} F_A= -\star D_{-A^{*}}\star\ F_A = -\star D_{A}\star F_{A}= \pm \star D_{A}F_{A}= 0.\]
The Minkowski metric case is similar.
\qed

SD and ASD solutions to Yang-Mills equations play important roles in Yang-Mills theory. Thus they warrant a dedicated name.
\begin{defn}
With $A$ being an $\mathcal{A}$-valued connection, if the associated curvature form $F_{A}$ is a nontrivial SD or ASD solution to Yang-Mills equations in vacuum, then it is said to be a Yang-Mills instanton solution.
\end{defn}

\begin{prop}
Suppose $\eta=A_1dz_1+A_2dz_2\in \Omega^1(\mathcal{A})$ is a normal Maurer-Cartan form and $A=\eta-\eta^*$. Then $F_A$ is an instanton solution to Yang-Mills equations under the Euclidean metric.
\end{prop}
\proof
In this case, using the Maurer-Cartan equation (\ref{MCartan}) we have 
\begin{align*}
F_A&=dA+A\wedge A\\
&=d\eta-d\eta^*+\eta\wedge \eta+\eta^*\wedge \eta^*-\eta\wedge \eta^*-\eta^*\wedge \eta\\
&=-\eta\wedge \eta^*-\eta^*\wedge \eta=-\sum_{1\leq i, j\leq 2}[A_j, A_k^*]dz_j\wedge d\bar{z}_k.
\end{align*}
Since $\eta$ is normal, we have 
\[F_A=-([A_1, A_2^*]dz_1\wedge d\bar{z}_2+[A_2, A_1^*]dz_2\wedge d\bar{z}_1),\]
which is ASD under the Euclidean metric by Proposition \ref{esdmsd}. The assertion then follows from Corollary \ref{SHsolution}.
\qed

\subsection{Pluriharmonic Connection Forms}
Recall that a smooth function $f$ on a domain $\Omega\in \C^n$ is said to be {\em pluriharmonic} if $\bar{\partial}\partial f=0$. And it is well-known that $f$ is pluriharmonic if and only if it is the real part of a holomorphic function on $\Omega$ \cite{Kr,range2013holomorphic}. The following definition is a natural extension of this idea to differential forms.
\begin{defn}
A connection form $A\in \Omega^1(\mathcal{A})$ is said to be pluriharmonic if $\bar{\partial}\partial A=0$.
\end{defn}
In the case $A$ is of the type $\eta-\eta^*$, where $\eta(z)=A_1(z)dz_1+A_2(z)dz_2\in \Omega^1(\mathcal{A})$, the components of the curvature form $F_A$ can be expressed as
\begin{align}
F_{{1}\overline{1}} &=-\partial_{1} A^*_{1} -\bar{\partial}_{1} A_{{1}} -\left[A_{{1}}, A^*_{1}\right],
\ F_{12} = \partial_{1} A_{2}- \partial_{2} A_{1} +\left[A_{1}, A_{2}\right], \hspace{3mm} \ F_{\overline{1}\overline{2}}=-(F_{12})^*, \label{eq:4.22}\\
F_{{2}\overline{2}} &=-\partial_{2} A^*_{2} -\bar{\partial}_{2} A_{{2}} -\left[A_{{2}}, A^*_{2}\right], \ F_{{1}\overline{2}} = -{\partial}_{1} A^*_{{2}} -\bar{\partial}_{2} A_{{1}} -\left[A_{{1}}, A^*_{{2}}\right], \  F_{2\overline{1}}=(F_{{1}\overline{2}})^*. \label{eq:4.23}
\end{align}
Clearly, if $\eta$ is holomorphic, the potential form $A=\eta-\eta^*$ is pluriharmonic and the above equations reduce to 
\begin{align}
F_{{1}\overline{1}} &=-\left[A_{{1}}, A^*_{1}\right],
\ F_{12} = \partial_{1} A_{2}- \partial_{2} A_{1} +\left[A_{1}, A_{2}\right], \ F_{\overline{1}\overline{2}}=-(F_{12})^*, \label{eq:4.24}\\
F_{{2}\overline{2}} &=-\left[A_{{2}}, A^*_{2}\right], \ F_{{1}\overline{2}} = -\left[A_{{1}}, A^*_{{2}}\right], \  F_{2\overline{1}}=(F_{{1}\overline{2}})^*.  \label{eq:4.25}
\end{align}

The simplest case is when $A_1$ and $A_2$ are constants. Then, the above equations are further reduced to:
\begin{align*}
F_{{1}\overline{1}} &=-\left[A_{{1}}, A^*_{1}\right],
\ F_{12} = \left[A_{1}, A_{2}\right], \hspace{4mm} \ F_{\overline{1}\overline{2}}=-(F_{12})^*, \\
F_{{2}\overline{2}} &=-\left[A_{{2}}, A^*_{2}\right], \ F_{{1}\overline{2}} = -\left[A_{{1}}, A^*_{{2}}\right], \  F_{2\overline{1}}=(F_{{1}\overline{2}})^*. 
\end{align*}
\noindent The following is an immediate consequence. 
\begin{prop}\label{const}
Assume $\eta=A_1dz_1+A_2dz_2$ is a constant $1$-form and $A=\eta-\eta^*$.

\begin{enumerate}[label=\textup{(\alph*)}, leftmargin=*]
\item With respect to the Euclidean metric,
\begin{align*}
F_{A} \textup{ SD} & \Leftrightarrow [A_1, A_2^*]=0,  [A_1, A_1^*]= [A_2, A_2^*], \\
F_{A} \textup{ ASD} & \Leftrightarrow  [A_1, A_2]=0, [A_1, A_1^*]+ [A_2, A_2^*]=0.
\end{align*}

\item With respect to the Minkowski metric, 
\begin{align*}
F_{A} \textup{ SD} & \Leftrightarrow \textup{$A_1, A_2$: normal and}\ [A_1^*-iA_1, A_2]=0,\\
F_{A} \textup{ ASD} & \Leftrightarrow \textup{$A_1, A_2$: normal and}\ [A_1^*+iA_1, A_2]=0.
\end{align*}
\end{enumerate}
\end{prop}
\proof Part (a) follows directly from Proposition \ref{esdmsd}. For part (b), since $[A_1, A_1^*]$ and $ [A_2, A_2^*]$ are both self-adjoint, the third equations in (\ref{eq:4.20}) and (\ref{eq:4.21}) imply that $[A_1, A_1^*]= [A_2, A_2^*]=0$, i.e., $A_1$ and $A_2$ are normal for both the SD and ASD case. Rewriting the first equation in (\ref{eq:4.20}) and resp. (\ref{eq:4.21}), we obtain the commutativity of $A_2$ with $A_1^*-iA_1$ and resp. $A_1^*+iA_1$.\qed

The case (b) above indicates that normal matrices play a natural role in the study of Yang-Mills equations. Indeed, this case has many solutions. For instance, matrices satisfying $A_1^*=iA_1$ (resp. $A_1^*=-iA_1$) and $A_2$ normal provide a particular class of SD (resp. ASD) nontrivial solutions. For earlier study on constant Yang-Mills curvature fields, we refer the reader to \cite{Sch88}. The following theorem extends Proposition \ref{const}.

\begin{thm}\label{thm:normal}
Assume $\eta(z)=A_1(z)dz_1+A_2(z)dz_2$ is holomorphic and normal and set $A=\eta-\eta^*$. 
\begin{enumerate}[label=\textup{(\alph*)}, leftmargin=*]
\item Under the Euclidean metric, the curvature $F_A$ is SD if and only if $A_1(z)$ and $A_2(z)$ commute for each $z\in \C^2$. In this case, 
\begin{equation}\label{SDE}
F_A=(\partial_1 A_2-\partial_2 A_1)dz_1\wedge dz_2-(\overline{\partial}_1 A^*_2-\overline{\partial}_2 A^*_1)d\bar{z_1}\wedge d\bar{z_2},\end{equation} which is constant $0$ on $\C^2$ if and only if $\eta$ is closed.

The curvature $F_A$ is ASD if and only if $F_{12}=0$. In this case, 
\begin{equation}\label{ASDE}
F_A=-[A_1, A^*_2]dz_1\wedge d\bar{z_2}-[A_2, A^*_1]dz_2\wedge d\bar{z_1}.\end{equation}
\item Under the Minkowski metric, $F_A$ is SD if and only if 
\begin{equation}\label{SDM}
[A_1^*-iA_1, A_2]=i(\partial_1A_2-\partial_2A_1).
\end{equation} It is ASD if and only if 
\begin{equation}\label{ASDM}
[A_1^*+iA_1, A_2]=-i(\partial_1A_2-\partial_2A_1).\end{equation}
\end{enumerate}
\end{thm}
\proof Since $\eta$ is holomorphic and normal, we have $F_{1\bar{1}}=F_{2\bar{2}}=0$ and $F_{1\bar{2}}=-[A_1, A_2^*]$. Hence, $F_A$ is SD if and only if $[A_1(z), A^*_2(z)]=0$ for each $z$. Since $A^*_2(z)$ is normal, the Fuglede theorem \cite{Fu} implies that $[A_1(z), A^*_2(z)]=0$ if and only if $A_1(z)$ commutes with $A_2(z)$. Further, since $d\eta=(\partial_1A_2-\partial_2A_1)dz_1\wedge dz_2$, we see that $F_A=0$ if and only if $d\eta=0$.

For part (b), $F_A$ is ASD if and only if $F_{12}=0$. In this case,
\[F_A=F_{1\bar{2}}dz_1\wedge d\bar{z_2}+F_{2\bar{1}}dz_2\wedge d\bar{z_1}
=-[A_1, A^*_2]dz_1\wedge d\bar{z_2}-[A_2, A^*_1]dz_2\wedge d\bar{z_1}.\]
In the Minkowski metric case, the first equation in (\ref{eq:4.20}) combined with (\ref{eq:4.24}) and (\ref{eq:4.25}) gives $[A_1^*, A_2]=i(\partial_1A_2-\partial_2A_1+[A_1, A_2])$, and the claim follows.
\qed

Applying $\bar{\partial_1}$ and $\bar{\partial_2}$ to equation (\ref{SDM}), we arrive at the following fact.
\begin{cor}
Suppose $A=\eta-\eta^*$ is given as in \textup{Theorem \ref{thm:normal}}. If $F_A$ is SD or ASD under the Minkowski metric, then $[\bar{\partial_j}A_1^*, A_2]=0, j=1, 2$.
\end{cor}

It is also worth describing the curvature form $F_A$ for Theorem \ref{thm:normal} (b) under the Minkowski metric. For the SD case, based on (\ref{eq:4.24}), (\ref{eq:4.25}), and (\ref{SDM}) we have 
\begin{align*}
F_A=F_{12}dz_1\wedge dz_2+F_{\bar{12}}d\bar{z}_1\wedge d\bar{z}_2+F_{1\bar{2}}dz_1\wedge d\bar{z}_2+F_{2\bar{1}}dz_2\wedge d\bar{z}_1=\omega_1-\omega_1^*,
\end{align*}
where $\omega_1=[A_2, A_1^*](idz_1\wedge dz_2-dz_2\wedge d\bar{z}_1)$. For the ASD case, similar computation based on (\ref{eq:4.24}), (\ref{eq:4.25}), and (\ref{ASDM}) gives $F_A=\omega_2-\omega_2^*$, where $\omega_2=-[A_2, A_1^*](idz_1\wedge dz_2+dz_2\wedge d\bar{z}_1)$.
%\iffalse%%%%%%%%%%%%%%%%%

\section{Pluriharmonic Solutions}

In the Euclidean ASD case of Theorem \ref{thm:normal}, the condition
$$
F_{12}=\partial_{1} A_{2}- \partial_{2} A_{1} +\left[A_{1}, A_{2}\right]=0
$$
in fact completely characterizes the solutions: the Yang-Mills field $F_A$ vanishes, and all components of $A$ commute, which means there is no nontrivial solutions in this case. Interestingly, in the Minkowski metric case there exist nontrivial pluriharmonic solutions. In the course of the proof we shall use the following lemma. To proceed, we let $\B(H)$ denote the set of bounded linear operators on a separable Hilbert space $H$. A subset $M\subset \B(H)$ is said to be commuting if $[A, B]=0$ for all $A, B\in M$. The {\em commutant} of $M$ is defined as \[M':=\{B\in \B(H)\mid [A, B]=0\ \forall A\in M\}.\]
\begin{lem}\lb{com} Consider $A_1, A_2 \in \B(H)$ such that $A_2$ is normal and $[A_1, A_2]$ commutes with $A_2$. Then $[A_1, A_2]=0$.
\end{lem}
\begin{proof} Let $E$ be the spectral measure on $\sigma(A_2)$. For any $z_0 \in \sigma(A_2)$ and $\epsilon>0$, we let $B(z_0,\epsilon):=\{z\in \C^2\mid \|z-z_0\|<\epsilon\}$ and set $S:=B(z_0, \epsilon) \cap \sigma(A_2)$, $P_{z_0}:=E(S)$. Then functional calculus gives
\begin{align*}
\|P_{z_0}(A_2-z_0)\| &= \left\|\int_{\sigma(A_2)}\chi_S(\xi)(\xi-z_0)dE(\xi)\right\|\\
&\leq \left\|\int_S |\xi-z_0|dE(\xi)\right\|\leq \epsilon.
\end{align*}

Since $[A_1, A_2]$ commutes with $A_2$, it also commutes with all its spectral projections $P$. Therefore,
$[A_1, A_2] P = P [A_1, A_2] P$. It follows that
$$\begin{aligned}
\|P_{z_0} [A_1, A_2] P_{z_0}\| &= \|P_{z_0} A_1 A_2 P_{z_0} -z_0P_{z_0} A_1 P_{z_0} +z_0 P_{z_0} A_1P_{z_0}-P_{z_0} A_2 A_1 P_{z_0}\| \\
&\leq \|P_{z_0} A_1 (A_2-z_0) P_{z_0}\| + \|P_{z_0} (A_2-z_0) A_1 P_{z_0}\| \leq 2\epsilon \|A_1\|.
\end{aligned}$$

Since $\sigma(A_2)$ is compact, the open cover $\{B(z_0, \epsilon) \mid z_0 \in \sigma(A_2)\}$ contains a finite subcover 
$\{B_k\mid 1 \leq k \leq N\}$, which we further refine by taking $\tilde B_{1}=B_1$ and $\tilde B_{k+1} = B_{k+1} \setminus (\cup_{j=1}^k B_j)$. Then the projections $P_k:=E(\tilde B_{k})$ are orthogonal and add up to the identity. Therefore,
\[[A_1, A_2]=\sum_{j=1}^N[A_1, A_2] P_k=\bigoplus_{k=1}^NP_k[A_1, A_2] P_k,\]
and hence 
\[\|[A_1, A_2]\| \leq \max_k \|P_k[A_1, A_2] P_k\| \leq 2\epsilon \|A_1\|.\]
Since $\epsilon$ can be arbitrarily small, we obtain $[A_1, A_2] = 0$.\qed
\end{proof}

\subsection{The Euclidean Metric Case}

In the Euclidean metric, Theorem 5.11 produces nontrivial abelian solutions in the SD case. It is thus a bit surprising to find out that the ASD case only has trivial solutions. This subsection presents a proof to this fact. 
\begin{prop}\lb{euclid}
Let $A_1, A_2: \C^2 \to \B(H)$ be holomorphic and normal that fulfill the identity
\be\lb{euclidean}
\partial_2 A_1 - \partial_1 A_2 = [A_1, A_2].
\ee
Then both sides of the identity vanish and $\ran A_1\cup \ran A_2$ is commuting.
\end{prop}

\begin{proof}
Since $A_1$ and $A_2$ are holomorphic, their adjoints $A_1^*$ and $A_2^*$ are conjugate holomorphic. In particular, $\ov \partial_k A_j = 0$ and $\partial_k A_j^* = 0$ for $1 \leq j, k \leq 2$, and hence
$$
\partial_2 [A_1, A_2^*] = [\partial_2 A_1, A_2^*] + [A_1, \partial_2 A_1^*] = [\partial_2 A_1, A_2^*].
$$
Using (\ref{euclidean}) and the fact that $\ran A_2$ is commuting (Lemma \ref{lem:comm}), we obtain
$$
\partial_2 [A_1, A_2^*] = [[A_1, A_2], A_2^*].
$$
By Jacobi's identity
$$
[[A_1, A_2], A_2^*] + [[A_2, A_2^*], A_1] + [[A_2^*, A_1], A_2] = 0,
$$
hence $[[A_1, A_2], A_2^*] = [[A_1, A_2^*], A_2]$. Thus
$$
\partial_2 [A_1, A_2^*] = [[A_1, A_2^*], A_2].
$$
\noindent This ordinary differential equation fulfilled by $[A_1, A_2^*]$ has holomorphic coefficients, so its solution must also be holomorphic as a function of $z_2$. The solution can be computed explicitly: setting the initial condition at $w_2$, then it can be verified that
$$
[A_1(z_1, z_2), A^*_2(z_1, z_2)] = \exp(-I_2(z_1, z_2)) [A_1(z_1, w_2), A^*_2(z_1, w_2)] \exp(I_2(z_1, z_2)),
$$
where $I_2(z_1, z_2) := \int_{w_2}^{z_2} A_2(z_1, \xi) \dd \xi$ is an antiderivative of $A_2$.

Since $[A_1, A_2^*]$ is holomorphic as a function of $z_2$, so are its higher derivatives $\partial_1^k \partial_2^\ell [A_1, A_2^*]$. It follows that 
$$
0=\ov\partial_2 \partial_1^k \partial_2^\ell [A_1, A_2^*] = [\partial_1^k \partial_2^\ell A_1, \ov\partial_2A_2^*]=[\partial_1^k \partial_2^\ell A_1, (\partial_2A_2)^*], \ \ k, \ell\geq 0,
$$
and consequently $[\partial_1^k \partial_2^\ell A_1, (\partial_2^m A_2)^*] = 0,\ m \geq 1.$ The Fuglede theorem gives
\be\lb{comm}
[\partial_1^k \partial_2^\ell A_1, \partial_2^m A_2] = 0.
\ee
Up to here, all commutators were computed point-wise. However, expanding $A_1$ and $A_2$ into Taylor series, we now obtain that
$$
[A_1(w), \partial_2A_2(z)] = 0,\ \ z, w\in \C^2.
$$
Setting $\K = (\ran A_1)' \cap (\ran A_2)'$, this shows $\ran \partial_2A_2\subset \K$. Since equation (\ref{euclidean}) is symmetric in $A_1$ and $A_2$, we likewise get $\ran \partial_1A_1\subset \K$. For all $z, w\in \C^2$, expanding $A_2$ as a power series in $z_2$, equation (\ref{comm}) gives 
\be\lb{expr'}
[A_1(z), A_2(z_1, w_2)-A_2(z)] = 0.
\ee
Due to the symmetry of (\ref{euclidean}), we also have $[A_1(z_1, z_2) - A_1(w_1, z_2), A_2(w)] = 0$. It follows that
\be\lb{expr}
[A_1(z), A_2(z)] = [A_1(z_1, z_2), A_2(z_1, w_2)] = [A_1(w_1, z_2), A_2(z_1, w_2)].
\ee
Differentiating (\ref{euclidean}) with respect to $z_1$, we obtain 
\be\lb{statement}
\partial_1 [A_1, A_2] = \partial_2 (\partial_1 A_1) - \partial_1^2 A_2 \in \mc (\ran A_2)'.
\ee
Moreover, by (\ref{expr}) we have $\partial_1 [A_1, A_2] = [A_1(w_1,z_2), \partial_1 A_2(z_1,w_2)]$. An application of Lemma \ref{com} gives $\partial_1 [A_1, A_2] = 0$. Symmetrically, we obtain $\partial_2 [A_1, A_2] = 0$. This indicates that the commutator $[A_1, A_2] = C$ is constant. Furthermore, the fact
$$
[A_1(w_1, z_2), \partial_1 A_2(z_1, w_2)] = 0, \ \ z, w\in \C^2,
$$
indicates that $\ran \partial_1 A_2 \subset \K$. By symmetry, $\partial_2 A_1 \in \K$, and it follows from (\ref{euclidean}) that $C\in \K$. Another application of Lemma \ref{com} gives $C=0$. Hence both sides of (\ref{euclidean}) vanish.

Integrating $\partial_2 A_1 = \partial_1 A_2$ in $z_2$, we obtain an expression for $A_1$ in terms of initial data $A_1(z_1, w_2)$ and the derivatives and integrals of $A_2$. This shows that $A_1(z_1, z_2)-A_1(z_1, w_2) \in \K$. Consequently, also using (\ref{expr'}), we obtain
$$\begin{aligned}
[A_1(z_1, z_2), A_2(w_1, w_2)] &= [A_1(z_1, w_2), A_2(w_1, w_2)] \\
&= [A_1(w_1, w_2), A_2(w_1, w_2)] = C = 0.
\end{aligned}$$
Thus all values of $A_1$ commute with all values of $A_2$, as claimed.\qed
\end{proof}

We return to the setting of a $C^*$-algebra $\A$ with a faithful tracial state $\Tr$ and let $\mathcal{H}$ stand for the Hilbert space completion of $\A$ with respect to the inner product $\langle a, b\rangle=\Tr (ab^*)$. The map $\phi: \A\to \B(\mathcal{H})$ defined by \[\phi(a)h=ah,\ \ a\in \A, h\in \mathcal{H}\] is known to be a faithful representation of $\A$. This process of defining $\phi$ is often referred to as the Gelfand-Naimark-Segal (GNS) construction \cite{davidson1996calgebras}. In the context of normal operators, it makes sense to consider unitary gauge transformations, which preserve normality.
\begin{cor} Let $\A$ be a $C^*$-algebra with a faithful tracial state and assume $A=\eta - \eta^*$ is an ASD solution to the vacuum Yang-Mills equations on a Euclidean background such that $\eta$ is holomorphic and normal. Then $F_A=0$, and $A$ is gauge equivalent to $0$ by means of a commuting unitary gauge transformation.
\end{cor}
\begin{proof} Due to representation $\phi$, without loss of generality we may assume $\A\subset \B(\mathcal{H})$. By Proposition \ref{euclid}, all values of $A_1$, $A_1^*$, $A_2$, and $A_2^*$ are normal and commute. In addition, the fact $\partial_2 A_1 = \partial_1 A_2$ indicates that the holomorphic $1$-form
$$
\eta = A_1 dz_1 + A_2 dz_2
$$
is closed. Since the domain $\C^2$ is simply connected, Poincar\'e lemma implies that $\eta$ has an antiderivative, unique up to a constant, given by the path integral
$$
h(z) = \int_{w}^{z} \eta.
$$
Then $g:=e^{h-h^*}$ commutes with $A$, and $A$ is the gauge transformation of $0$ because
$$
g^{-1}dg = dh - dh^* = \eta-\eta^*=A.
$$
Finally, $g$ is unitary because $g^* = e^{h^*-h} = g^{-1}$. Note that $\log g$ is pluriharmonic, but $g$ itself need not be.\qed
\end{proof}

\subsection{The Minkowski Metric Case}

On the other hand, the vacuum Yang-Mills equations have nontrivial pluriharmonic solutions in the Minkowski metric, in both the SD and the ASD cases. Since the SD and ASD conditions in Theorem \ref{thm:normal} (b) are similar, the proof is parallel. In the sequel, we shall only study the SD case. In fact, we will show that the SD condition
$$
[A_1- iA_1^*, A_2]= i(\partial_1A_2-\partial_2A_1).
$$
leads to a complete characterization of the solutions. 
\begin{prop}\label{QQ}
Assume $A_j: \C^2 \to \B(H), j =1, 2$, are holomorphic and normal, and they fulfill the identity
\be\lb{minkowski}
i(\partial_2 A_1 - \partial_1 A_2) = [A_1-iA_1^*, A_2].
\ee
Then both sides of the identity vanish. Moreover, $\ran(A_1-iA_1^*)\cup \ran A_2$ is a commuting set, and the commutator $[A_1(z), A_2(w)]$ only depends on $w_2$.
\end{prop}
\begin{proof}
In equation (\ref{minkowski}), all terms are holomorphic, except possibly for $[A_1^*, A_2]$, which then must be holomorphic as well. Therefore
$$
0 = \ov \partial_1^k\ov\partial_2^\ell [A_1^*, A_2] = [\ov  \partial_1^k\ov\partial_2^\ell  A_1^*, A_2] = [(\partial_1^k\partial_2^\ell A_1)^*, A_2],\ \ k, \ell \geq 0, k+\ell\geq 1.
$$
Then Fuglede theorem implies $[\partial_1^k \partial_2^\ell A_1, A_2] = 0$, which in particular shows that both left-hand side terms $\partial_2 A_1$ and $\partial_1 A_2$ in (\ref{minkowski}) commute pointwise with $A_2$. At the same time, the right-hand side is a commutator involving $A_2$. Lemma \ref{com} then asserts that both sides of (\ref{minkowski}) vanish.

Moreover, expanding $A_1$ into a Taylor series, we have
\be\lb{identity1}
[A_1(z)-A_1(w), A_2(w)] = 0=[A^*_1(z)-A^*_1(w), A_2(w)] ,\ \ z, w\in \C^2,
\ee
which implies that 
\[[A_1(z)-iA_1^*(z), A_2(w)]=[A_1(w)-iA_1^*(w), A_2(w)]=0.\] This shows that $\ran(A_1-iA_1^*)\cup \ran A_2$ is a commuting set.

Furthermore, (\ref{identity1}) also indicates that $[A_1(z), A_2(w)]$ is independent of the choice of $z\in \C^2$.
Fixing any $\tilde{w}_1\in \C$ and using the fact $\partial_2 A_1 = \partial_1 A_2$, we have 
\[A_2(w_1, w_2)-A_2(\tilde{w}_1,w_2)=\int_{\tilde{w}_1}^{w_1}\partial_1 A_2(\xi, w_2)d\xi=\int_{\tilde{w}_1}^{w_1}\partial_2 A_1(\xi, w_2)d\xi,\]which belongs in $\K$. Thus, $[A_1(z), A_2(w)]=[A_1(z), A_2(\tilde{w}_1,w_2)]$, which depends only on $w_2$.
\end{proof}

We return to the case where $\A$ is a $C^*$ algebra with a faithful tracial state. Again, due to the GNS construction, we may assume $\A\subset \B(\mathcal{H})$.
\begin{thm}
The vacuum Yang-Mills equations have nontrivial pluriharmonic normal instanton solutions in the Minkowski metric.
Moreover, up to unitary gauge equivalence, such solutions are of the form $A'=\eta' - {\eta'}^*$, where $\eta'=A'_1dz_1+A'_2dz_2$ with $A'_1=i{A'_1}^*$ being constant and $A'_2$ being holomorphic, normal, and dependent only on $z_2$.
\end{thm}

\begin{proof} The existence of such solutions is shown in Example \ref{Masd} below. Here, we prove the remaining claims of the theorem.

First, we decompose $A_1$ as
$$
A_1 = \frac 1 2 ((A_1-iA_1^*)+(A_1+iA_1^*)) =: B_1+B_2,
$$
where, due to the fact $\ov{i^{1/2}} = i^{7/2} = -i^{3/2}$, 
$$
(i^{1/2} B_1)^* = (i^{1/2} A_1 - i^{3/2} A_1^*)^* = -i^{3/2} A_1^* + i^{1/2} A_1 =i^{1/2} B_1
$$
and
$$
(i^{1/2} B_2)^* = (i^{1/2} A_1 + i^{3/2} A_1^*)^* = -i^{3/2} A_1^* - i^{1/2} A_1 = -i^{1/2} B_2.
$$

Fix a $w\in \C^2$, Proposition \ref{QQ} and its proof indicates that $A_1(z)-iA_1^*(z), A_1(z)-A_1(w)$, and  $A_2(z_1,z_2)-A_2(w_1,z_2)$ are all in $\K$ for each $z\in \C^2$. The only components that may not commute with everything else are $B_{2}(w)$ and $A_2(w_1,z_2)$. We set $\tilde A_1(z)=A_1(z)-B_{2}(w)$, $\tilde A_2(z)=A_2(z)-A_{2}(w_1,z_2)$, and $\tilde \eta=\tilde A_1 dz_1 + \tilde A_2 dz_2$. Then the coefficients of $\tilde \eta$ are all in $\K$. Moreover, since $\partial_2 \tilde A_1 = \partial_2 A_1$, $\partial_1 \tilde A_2 = \partial_1 A_2$, and $\partial_2 A_1 = \partial_1 A_2$ due to Proposition \ref{QQ}, the $1$-form $\tilde \eta$ is exact.
Let $\tilde h$ be an antiderivative of $\tilde \eta$ and define $g=e^{-\tilde h+\tilde h^*}$. Recall that $\eta=A_1dz_1+A_2dz_2$ and $A=\eta-\eta^*$. Then the gauge transformation 
\[ A_g=g^{-1}Ag+g^{-1}dg=A+d(-\tilde h+\tilde h^*)=A-\tilde A= A'=\eta' - {\eta'}^*,\] where \[\eta'=\eta-\tilde \eta=B_{2}(w)dz_1+A_{2}(w_1,z_2)dz_2.\]\qed
\end{proof}

\begin{example}\label{Masd}
Consider any holomorphic function $h(z)$ on $\C^2$ and set 
\[U = \frac{1}{\sqrt{2}}\begin{pmatrix}
1 & 1\\
i & -i
\end{pmatrix}, \hspace{1cm}
H(z_1, z_2):=\int_0^{z_1}\partial_2h(\xi, z_2)d\xi,\ \ z_1, z_2\in \C.\] We define
\[A_1=\begin{pmatrix}
h & 0\\
0 & h+1-i
\end{pmatrix}, \ \ \ \text{and}\ \ \ A_2=U^* \begin{pmatrix}
H+\cos z_2 & 0\\
0 & H+i\sin z_2
\end{pmatrix}U.\]
Then $A_1$ and $A_2$ are holomorphic and normal in $\C^2$. Direct computation verifies that
\[[A_1^*-iA_1, A_2]=0=i(\partial_1 A_2-\partial_2A_1),\] showing that $A_1$ and $A_2$ satisfy \textup{Theorem \ref{thm:normal}} for the SD case with Minkowski background. Furthermore, 
\[[A_2, A_1^*]=\frac{1}{\sqrt{2}}e^{-i(z_2-\frac{\pi}{4})}\begin{pmatrix}
0 & 1\\
-1 & 0
\end{pmatrix},\]
indicating that the curvature field $F_A$ is nonzero according to the comments following \textup{Corollary 5.12}. This example can be modified slightly to suit the ASD case.
\end{example}
Observe that the solution in this example is gauge equivalent to the special case $h=0$. thus the choice of $h$ does not affect the curvature $F_A$. On the other hand, the choice of the unitary matrix $U$ matters. It will be interesting to describe the type of unitary matrices which permit a nontrivial solution in this case. It is also worth noting that the field $F_A$  does not depend on the variable $z_1$ whose real part is the time variable $x_0$, implying that $F_A$ is stationary.

It is important to note that the solutions provided in the preceding discussions are not the only kind of pluriharmonic instantons. The following example, whose curvature field is sometimes referred to as a quantized Dirac monopole, presents a different kind of pluriharmonic instantons.
\begin{example}\label{dirac}
Consider the operator-valued skew-Hermitian form $A(z)=\eta\left(z\right)-\eta^{*}\left(z\right)$,
where $\eta\left(z\right)=\bar{z}_{1}B_{1}dz_{1}+\bar{z}_{2}B_{2}dz_{2}$, and $B_{j}\in\mathcal{A}, j=1, 2$.
The curvature field $F_{A}$ in this case is sometimes referred to as a quantized Dirac monopole. Under the Euclidean metric, direct computation \cite[Proposition 6.1]{munshi2020maxwell} verifies that
\begin{enumerate}[label=\textup{\alph*)}, leftmargin=*]
\item $F_{A}$ \textup{SD} $\Leftrightarrow$ $\left[B_{1},B_{1}^{*}\right]=\left[B_{2},B_{2}^{*}\right]=\left[B_{1},B_{2}^{*}\right]=0$, and $B_{1}+B_{1}^{*}=B_{2}+B_{2}^{*}$;
\item $F_{A}$ \textup{ASD} $\Leftrightarrow$ $\left[B_{1},B_{1}^{*}\right]=\left[B_{2},B_{2}^{*}\right]=0$ and $B_{1}+B_{1}^{*}=-\left(B_{2}+B_{2}^{*}\right)$.
\end{enumerate}
\end{example}
\noindent Evidently, in both cases a) and b) above, the operators $B_1$ and $B_2$ are normal. Moreover, case a) implies that $B_1$ commutes with $B_2$ \cite{Fu}. A simple special case of b) is when
$\mathcal{A}=M_2(\C)$ with
\begin{equation*}
B_{1} = \begin{pmatrix}
-1 & 1\\
-i & -i
\end{pmatrix},\ \ \
B_{2} = \begin{pmatrix}
1 & -i\\
-1 & -i
\end{pmatrix},
\end{equation*}
in which the curvature field $F_{A}$ is nonzero. The next corollary extends this special case.
\begin{cor}
Let $\eta\left(z\right)=\bar{z}_{1}B_{1}dz_{1}+\bar{z}_{2}B_{2}dz_{2}$, where $B_{1}\in {\mathcal A}$ is normal with $B_{1}+B_{1}^{*} \neq 0$ and $B_{2}= -B_{1}^{*}$. Then $F_{A}$ is a Yang-Mills instanton under the Euclidean metric.
\end{cor}

\section{On A Classical Yang-Mills Instanton}

Analogous to the study of Maxwell's equations, this section considers the ``electric" and ``magnetic" components of Yang-Mills fields. We will also take a closer look at the classical Belavin-Polyakov-Schwartz-Tyupkin instanton and show that it is in fact skew-Hermitian and almost pluriharmonic.

\subsection{The Nonabelian Electric and Magnetic Components}

For Maxwell's equations (\ref{eq:4.1}), the Faraday 2-form (\ref{eq:4.2}) defines the electric component ${\bf E}$ and the magnetic component ${\bf B}$. For a $C^*$-algebra-valued complex $2$-form $F_A$ as in (\ref{eq:4.6}), comparing (\ref{eq:4.2}) with (\ref{eq:4.6}), we define ${\bf E}$ and ${\bf B}$ with the following corresponding components \cite{munshi2022complex}:
\begin{align}
E_{1} &= 2iF_{1\overline{1}}, \ \ E_{2} = -\left(F_{12}+F_{\overline{1}\overline{2}}+F_{1\overline{2}}-F_{2\overline{1}}\right), \ \ E_{3} &= -i\left(F_{12}-F_{\overline{1}\overline{2}}-F_{1\overline{2}}-F_{2\overline{1}}\right), \nonumber \\
B_{1} &= 2iF_{2\overline{2}}, \ \ B_{2} = -\left(F_{12}+F_{\overline{1}\overline{2}}-F_{1\overline{2}}+F_{2\overline{1}}\right), \ \ B_{3} &= -i\left(F_{12}-F_{\overline{1}\overline{2}}+F_{1\overline{2}}+F_{2\overline{1}}\right). \label{eq:4.23}
\end{align}
${\bf E}$ and respectively ${\bf B}$ are sometimes referred to as {\em chromo-electric} and {\em chromo-magnetic fields}. It is worth mentioning that, due to equations (\ref{gaugeinv}), for any gauge transformation $A\to A_g$, we have 
\[{\bf E}_{A_g}=g^{-1}{\bf E}_{A}g,\hspace{2cm} {\bf B}_{A_g}=g^{-1}{\bf B}_{A}g.\]
The following are direct consequences of Proposition \ref{esdmsd}.

\begin{cor}\label{magelec}
Consider a connection form $A\in \Omega^1(\mathcal{A})$.
\begin{enumerate}[label=\textup{(\alph*)}, leftmargin=*]
\item Under the Euclidean metric, $F_A$ SD $\Leftrightarrow {\bf B}={\bf E}$; $F_A$ ASD $\Leftrightarrow {\bf B}=-{\bf E}$.

\item Under the Minkowski metric, $F_A$ SD $\Leftrightarrow {\bf B}=i{\bf E}$; $F_A$ ASD $\Leftrightarrow {\bf B}=-i{\bf E}$.
\end{enumerate}
\end{cor}

A tuple ${\bf T}=(T_1, T_2, T_3)$ of operators in $\mathcal{A}$ is said to be {\em Hermitian} if $T_j$ is Hermitian for each $j$. If ${\bf T}:\C^2\to \mathcal{A}^3$ is an operator-valued vector function, then we say ${\bf T}$ is Hermitian if ${\bf T}(z)$ is Hermitian for every $z\in \C^2$. The following observation is an immediate consequence of Corollary \ref{magelec}.
\begin{cor}
There exists no instanton solution to Yang-Mills equations over $\C^2$ under the Minkowski metric for which both ${\bf E}$ and ${\bf B}$ are Hermitian.
\end{cor}

Then the inner product $\langle \mathbf{E}, \mathbf{B} \rangle$ can be calculated as
\begin{align}
\langle \mathbf{E}, \mathbf{B} \rangle&= \text{Tr}\left(E_1{B}^{*}_1+E_2{B}^{*}_2+E_3{B}^{*}_3\right) \nonumber\\
& = \text{Tr}\left(4F_{1\overline{1}}F_{2\overline{2}}^{*}+ 2\left(F_{12}-F_{2\overline{1}}\right)\left(F_{12}+F_{2\overline{1}}\right)^{*} +2\left(F_{\overline{1}\overline{2}}+F_{1\overline{2}}\right)\left(F_{\overline{1}\overline{2}}-F_{1\overline{2}}\right)^{*}\right).
\label{eq:4.31}
\end{align}
Then Corollary \ref{magelec} leads to the following facts.
\begin{cor}\label{EB}
Let $F_A$ be an instanton solution to Yang-Mills equations. 
\begin{enumerate}[label=\textup{(\alph*)}, leftmargin=*]
\item Under the Euclidean metric, $F_A$ is SD $\Leftrightarrow \langle \mathbf{E}, \mathbf{B} \rangle>0$; $F_A$ is ASD $\Leftrightarrow \langle \mathbf{E}, \mathbf{B} \rangle<0$.

\item Under the Minkowski metric, $F_A$ is SD $\Leftrightarrow$ $\langle \mathbf{E}, \mathbf{B} \rangle\in i{\mathbb R}_{\geq 0}$; $F_A$ is ASD $\Leftrightarrow \langle \mathbf{E}, \mathbf{B} \rangle \in -i{\mathbb R}_{\geq 0}$.
\end{enumerate}
\end{cor}

\begin{example}
If $F_A$ is the ASD curvature form given in \textup{(\ref{ASDE})} under the Euclidean metric, then direct computation gives
\begin{align*}
\langle \mathbf{E}, \mathbf{B} \rangle&=-2\Tr([A_1, A_2^*][A_2, A_1^*])+[A_2, A_1^*][A_1, A_2^*])=-4\Tr([A_1, A_2^*][A_2, A_1^*]).
\end{align*} Furthermore, if $F_A=\omega_2-\omega_2^*$ is the ASD curvature form given preceding \textup{Example \ref{dirac}} under the Minkowski metric, then direct computation gives
\[\langle \mathbf{E}, \mathbf{B} \rangle=4i\Tr([A_2, A_1^*][A_1, A_2^*]-[A_1, A_2^*][A_2, A_1^*])=0.\]
\end{example}

\subsection{The Belavin-Polyakov-Schwartz-Tyupkin Solution}

A classical instanton solution to Yang-Mills equations on $\mathbb{R}^{4} \simeq \C^2$ is the so-called {\em Belavin-Polyakov-Schwartz-Tyupkin} (BPST) solution \cite{belavin1979quantum, belavin1975pseudoparticle, ritter2003gauge}. We let $\mathbf{SU}\left(2\right)$ be the {\em special unitary group}, i.e.,  group of $2\times 2$ unitary matrices with determinant $1$. The BPST solution is given by the Maurer-Cartan form associated with the map
 $\gamma: \mathbb{R}^{4}\rightarrow\mathbf{SU}\left(2\right)$ defined by
\begin{equation}
\gamma\left(x\right)= \frac{1}{|x|}\left(x_{0}-i\sum_{j=1}^3x_{j}\sigma_{j}\right),
\label{eq:5.1}
\end{equation}
where $|x|=(x_0^2+x_1^2+x_2^2+x_3^2)^{1/2}$ and 
\[
\sigma_{1}= \begin{pmatrix}
0 & 1\\
1 & 0
\end{pmatrix},
\sigma_{2}= \begin{pmatrix}
0 & -i\\
i & 0
\end{pmatrix},
\sigma_{3}= \begin{pmatrix}
1 & 0\\
0 & -1
\end{pmatrix},
\] are called the {\em Pauli matrices}.
\begin{thm}\label{BPST}
The connection form $A=\frac{|x|^2}{|x|^{2}+\mu}\gamma^{-1} d\gamma$, for some constant $\mu>0$, gives a Yang-Mills instanton solution $F_A$ under the Euclidean metric. 
\end{thm}
This solution fits well in the complex variable setting developed here. To give a proof to the theorem, we start with a connection of the form $A\left(x\right):= f\left(x\right)\gamma^{-1}d\gamma$, where $f$ is a real-valued function. In light of Lemma \ref{skew}, the connection $A$ is skew-Hermitian. Thus, Proposition \ref{SHsolution} indicates that the associated curvature form $F_A$ is a solution to Yang-Mills equations if it is SD or ASD. In terms of the complex variables $z=\left(z_{1}, z_{2}\right)\in\mathbb{C}^{2}$ with $z_{1}=x_{0}+ix_{1}, z_{2}=x_{2}+ix_{3}$, we write
\begin{equation}
\gamma^{-1}d\gamma= \frac{-1}{2|z|^{2}}\left(\eta\left(z\right)-\eta^{*}\left(z\right)\right),
\label{eq:5.4}
\end{equation}
where \[\eta\left(z\right)=\left(\bar{z}_{1}d{z}_{1}-\bar{z}_{2}d{z}_{2}\right)\sigma_{1} -i\left(\bar{z}_{2}d{z}_{1}-\bar{z}_{1}d{z}_{2}\right)\sigma_{2} +\left(\bar{z}_{2}d{z}_{1}+\bar{z}_{1}d{z}_{2}\right)\sigma_{3}=:\sum_{j=1}^3\omega_j\sigma_j.\] The following observations are notable.
\begin{enumerate}[label=\textup{\arabic*)}, leftmargin=*]
\item $\omega_j, j=1, 2, 3,$ are solutions to Maxwell's equations \cite[Theorems 4.4, 4.15]{munshi2022complex}.

\item $\eta$ is pluriharmonic.

\item $d\eta$ is skew-Hermitian, and it is ASD under the Euclidean metric.
 \end{enumerate}
\noindent Claim 3) is not obvious, but it can be verified by direct computation. From this claim and (\ref{eq:5.4}), it follows that
\begin{align*}
d\left(\gamma^{-1}d\gamma\right) &= d\gamma^{-1}\wedge d\gamma \\
 &= d\left(\frac{-1}{2|z|^{2}}\right)\wedge \left(\eta -\eta^{*}\right) +\frac{-1}{2|z|^{2}}\left(d\eta -d\eta^{*}\right) \\
 &= d\left(\frac{-1}{2|z|^{2}}\right)\wedge \left(\eta -\eta^{*}\right) +\frac{-1}{|z|^{2}}d\eta.
\end{align*}
The associated curvature form is then
\begin{align}
F_{A} &= dA +A\wedge A \label{eq:5.5} \\
 &= df\wedge \gamma^{-1}d\gamma +f\left(d\gamma^{-1}\wedge d\gamma\right) +f^{2}\gamma^{-1}d\gamma\wedge \gamma^{-1}d\gamma. \nonumber
\end{align}
Now recall that $d\gamma^{-1}=-\gamma^{-1}d\gamma \cdot \gamma^{-1}$. This allows us to rewrite and simplify (\ref{eq:5.5}) as
\begin{equation}\label{bpst}
F_{A}= \left(\frac{-1}{2|z|^{2}}df +\left(f-f^{2}\right)d\left(\frac{-1}{2|z|^{2}}\right) \right)\wedge \left(\eta -\eta^{*}\right)+\frac{f-f^{2}}{|z|^{2}}d\eta.
\end{equation}
Since $d\eta$ is ASD, in order for $F_{A}$ to be ASD, we may set
\begin{equation}
\frac{1}{2|z|^{2}}df +\left(f-f^{2}\right)d\left(\frac{1}{2|z|^{2}}\right) =0.
\label{eq:5.6}
\end{equation}
Setting $\lambda=\frac{1}{2|z|^{2}}$, the equation becomes
\begin{equation}
\lambda df=(f^2-f)d\lambda.
\label{eq:5.2}
\end{equation}
Taking into account the convergence of the connection form $A$ as $|z|\to \infty$, we prefer that $f$ satisfies $|f(z)|\leq 1$ when $|z|$ becomes large. Then (\ref{eq:5.2}) gives $f\left(z\right)=\frac{|z|^2}{|z|^{2}+\mu}$ for some positive constant $\mu$, completing the proof of Theorem \ref{BPST}.

Equations (\ref{bpst}) and (\ref{eq:5.6}) indicate that the solution in Theorem \ref{BPST} is 
\begin{equation}\label{BPST2}
F_A= \frac{f-f^{2}}{|z|^{2}}d\eta=\frac{\mu}{(|z|^2+\mu)^2}d\eta=: p(z)d\eta.
\end{equation} Two more facts are worth noting:

\noindent 4) $F_A$ is skew-Hermitian due to fact 3), conforming with Proposition \ref{SHsolution}.

\noindent 5) $F_A$ converges to $0$ as $z\to \infty$, showing that it is well-defined on $\C^2\cup \{\infty\}$ which can be identified with the $4$-dimensional unit sphere $S^4$. This fact, in part, motivated the interest of studying Yang-Mills equations on $4$-dimensional manifolds.

Finally, we take a look at the ``electric'' and ``magnetic'' components of the BPST instanton solution. Using the function $p$ defined in (\ref{BPST2}), we write
\begin{equation*}
F_{A} = p\big(\sigma_{1}\left(dz_{1}\wedge d\bar{z}_{1}- dz_{2}\wedge d\bar{z}_{2}\right)+i\sigma_{2}\left(dz_{2}\wedge d\bar{z}_{1}- dz_{1}\wedge d\bar{z}_{2}\right)+\sigma_ {3}\left(dz_{2}\wedge d\bar{z}_{1}+ dz_{1}\wedge d\bar{z}_{2}\right)\big).
\end{equation*}
It follows that $F_{12}= F_{\overline{1}\overline{2}}=0$ and 
\begin{align}\label{eq:5.7}
F_{1 \overline{1}} = p\sigma_{1}, \ \ F_{1 \overline{2}} = p\left(-i\sigma_{2}+ \sigma_{3}\right),\ \ 
F_{2 \overline{1}} = p\left(i\sigma_{2}+ \sigma_{3}\right),\ \ F_{2 \overline{2}} &= -p\sigma_ {1}.
\end{align}
Then, in view of (\ref{eq:4.23}), the components of $\mathbf{E}$ and $\mathbf{B}$ are given by
\begin{equation}
E_{j}= 2ip\sigma_{j}, \text{  } B_{j} =-2ip\sigma_{j},\ \ \ j= 1, 2, 3.
\label{eq:5.8}
\end{equation}
Since $E_{j}B_{j}^{*}=-4p^{2}\sigma_{j}\sigma_{j}^{*}=-4p^2$, it follows that 
\begin{align} \label{eq:5.9} 
\langle \mathbf{E}, \mathbf{B} \rangle &= -4p^{2}\text{Tr}\left(\sum_{j=1}^{3}\sigma_{j}^{2}\right) \nonumber \\
&= -4p^{2}\text{Tr}\left(3I_{2}\right)= -24p^{2}. \nonumber
\end{align}
Note that this inner product is negative for nonzero $z\in\mathbb{C}^{2}$, conforming with Corollary \ref{EB}.

\section{Concluding Remarks}

Complex analysis is one of the most useful tools in mathematics, and it is playing an increasingly important role in modern physics. It is thus meaningful to reinterpret some fundamental theories in physics from a complex variable perspective, for instance special relativity, Maxwell's equations, and Yang-Mills theory whose original formulations were in real variables. 
 This paper is a non-commuting generalization of the approach in \cite{munshi2022complex} on Maxwell's equations. It is a part of the authors' ongoing project to explore greater applications of complex analysis to physics theories.\\

{\bf Acknowledgments.} This paper is in part based on the second-named author's doctoral dissertation submitted to the University at Albany. He is grateful to the Department of Mathematics and Statistics for providing him an opportunity to pursue his research interests.

\end{document}